\documentclass{article}
\def\IEEEversion{0}
\usepackage[a4paper,hmargin=3cm,vmargin=4.5cm]{geometry}
\usepackage[utf8]{inputenc}
\usepackage{cmap}

\usepackage[english]{babel}

\usepackage{amsmath}
\usepackage{amssymb}
\usepackage{amsthm}
\usepackage{dsfont}
\usepackage{mathrsfs}

\usepackage{hyperref}

\usepackage{booktabs}
\usepackage{multirow}
\usepackage{multicol}
\usepackage[table]{xcolor} 

\usepackage{graphicx} \usepackage[export]{adjustbox} \usepackage{framed}

\usepackage{enumerate}

\usepackage{titlesec} \if\IEEEversion0
\usepackage{sectsty}  \allsectionsfont{\normalfont\sffamily\bfseries}
\fi

\usepackage{algorithm}
\usepackage{algpseudocode}
\algnewcommand{\LineComment}[1]{\State{\color{gray} \(\triangleright\) #1}}
\usepackage{listings}

\usepackage{subcaption}
\usepackage{graphicx}
\usepackage{tikz}
\usetikzlibrary{decorations.pathreplacing}
\usetikzlibrary{arrows}
\usetikzlibrary{arrows.meta}
\usetikzlibrary{patterns}
\usetikzlibrary{shapes}
\usetikzlibrary{shadows}
\usetikzlibrary{calc}
\usetikzlibrary{math}
\tikzstyle{fun}=[draw,very thick,fill=white]
\tikzstyle{fork}=[shape=circle,inner sep=0pt,minimum size=3pt,fill=black]
\tikzstyle{op}=[inner sep=2pt]

\usepackage{ifthen}
\usepackage{comment}

\usepackage{array}
\newcolumntype{t}{>{\tt}c}

\newtheorem{theorem}{Theorem}
\newtheorem{definition}[theorem]{Definition}
\newtheorem{lemma}[theorem]{Lemma}
\newtheorem{corollary}[theorem]{Corollary}

\newtheorem{proposition}[theorem]{Proposition}

\newtheorem{problem}{Problem}

\newtheorem{notation}[theorem]{Notation}

\theoremstyle{remark}
\newtheorem{remark}[theorem]{Remark}
\newtheorem{example}[theorem]{Example}

\newcommand\LL{\mathcal{L}}

\newcommand\F{\mathbb{F}}
\newcommand\ftwo{\mathbb{F}_2}
\newcommand\field[1]{\mathbb{F}_{2^{#1}}}

\newcommand\bigoh[1]{\mathrm{O}\left( #1  \right)}
\newcommand\trace{\mathrm{Tr}}

\newcommand\scalarprod[2]{#1 \cdot #2}
\newcommand\walsh[1]{\mathcal{W}_{#1}}

\newcommand\openbutterfly[3]{\mathsf{H}^{#1}_{#2, #3}}

\newcommand\rank{\mathsf{rank}}

\newcommand\walshZeroes[1]{\mathcal{Z}_{#1}}

\newcommand\spaceInput{\mathcal{V}}
\newcommand\spaceOutput{\mathcal{V}^{\perp}}

\newcommand\codebook[1]{\Gamma_{#1}}

\newcommand\jacobian[2]{\mathrm{Jac}\ \! #1 (#2)}
\newcommand\Jlin[2]{\mathrm{Jac}_{\rm lin}\ \! #1 (#2)}
\newcommand\JacPoly{\mathscr{J}\!}
\newcommand\partialDiff[2]{\frac{\partial #1}{\partial #2}}

\newcommand{\rktable}[1]{\mathcal{R}(#1)}
\newcommand{\rkdist}[1]{\mathcal{R}_{\mathrm{dist}}(#1)}
\newcommand{\ddist}[1]{\mathcal{D}_{\mathrm{dist}}(#1)}

\renewcommand{\leq}{\leqslant}
\renewcommand{\geq}{\geqslant}

\if\IEEEversion1
\usepackage{cite}
\fi 
\title{Recovering or Testing Extended-Affine Equivalence}
\author{Anne Canteaut, Alain Couvreur, L\'eo Perrin}
\date{\today}
\begin{document}

\maketitle

\begin{abstract}
  Extended Affine (EA) equivalence is the equivalence relation between
  two vectorial Boolean functions $F$ and $G$ such that there exist
  two affine permutations $A$, $B$, and an affine function $C$
  satisfying $G = A \circ F \circ B + C$. While the problem has a
  simple formulation, it is very difficult in practice to test whether
  two functions are EA-equivalent.  This problem has two variants:
  {\em EA-partitioning} deals with partitioning a set of functions into disjoint
  EA-equivalence classes, and \emph{EA-recovery} is about recovering
  the tuple $(A,B,C)$ if it exists.

  In this paper, we present a new algorithm that efficiently solves
  the EA-recovery problem for quadratic functions. Although its
  worst-case complexity occurs when dealing with APN functions, it
  supersedes, in terms of performance, all previously known algorithms
  for solving this problem for all quadratic functions and in any
  dimension, even in the case of APN functions. This approach is based
  on the Jacobian matrix of the functions, a tool whose study in this
  context can be of independent interest.

  The best approach for EA-partitioning in practice mainly relies on class
  invariants. We provide an overview of the known invariants along
  with a new one based on the \emph{ortho-derivative}. This new
  invariant is applicable to quadratic APN functions, a specific
  type of functions that is of great interest, and of which tens of
  thousands need to be sorted into distinct EA-classes. Our
  ortho-derivative-based invariant is very fast to compute, and it
  practically always distinguishes between EA-inequivalent quadratic APN functions.
\end{abstract}
 \section{Introduction}
\label{sec:intro}
Nonlinear vectorial Boolean functions (also known as {\em Sboxes}) are crucial
building-blocks of most symmetric cryptosystems. For instance, their
properties can be used to prove that a cryptographic primitive is safe from
differential and linear cryptanalysis~\cite{FSE:Nyberg12}.  The search for
vectorial Boolean functions that guarantee an optimal resistance to
these attacks has then motivated a long line of research in the last
30~years. Indeed, establishing a list (as complete as possible) of
optimal functions that could be used as Sboxes would be very
helpful to designers of cryptographic primitives.

However, the huge number of vectorial Boolean functions from
\(\F_2^n\) into \(\F_2^m\), even for small~\(n\) and \(m\), makes an exhaustive
search infeasible. In order to help in this task, the
functions are then considered up to some equivalence relations which
preserve the considered cryptographic properties.
The most natural notions of equivalence are affine equivalence, and
its generalization called {\em extended affine equivalence}, which are
defined as follows.

\begin{definition}[(Extended) Affine Equivalence]
  Two vectorial Boolean functions $F : \ftwo^{n} \to \ftwo^{m}$ and $G :
  \ftwo^{n} \to \ftwo^{m}$ are \emph{affine equivalent} if $G = A \circ F
  \circ B$ for some affine permutations $A$ of $\ftwo^{n}$ and $B$ of
  $\ftwo^{m}$. They are \emph{extended affine equivalent} (EA-equivalent)
  if $G = A\circ F \circ B + C$ where $A$ and $B$ are as before and
  where $C : \ftwo^{n} \to \ftwo^{m}$ is an affine function.
\end{definition}

Determining whether two functions
\(F,G: \F_2^n \rightarrow \F_2^m\)
are affine or EA-equivalent is a problem that appears in several situations:
when classifying vectorial Boolean functions, but also in
cryptanalysis, where several types of attacks need to recover, if it
exists, the transformation between two equivalent functions
(see e.g.~\cite{EC:Dinur18} for details).
This last situation then corresponds to the following problem,
named {\em EA-recovery}, which is investigated in this paper.

\begin{problem}[EA-recovery]\label{problem:ea-recovery}
  Let $F$ and $G$ be two functions from $\ftwo^{n}$ into
  $\ftwo^{m}$. Find, if they exist,
two affine permutations \(A\) and \(B\), and an affine function \(C\) such that \(G = A \circ F \circ B + C\).
\end{problem}

But, in practice, the following variant of
Problem~\ref{problem:ea-recovery} better captures some situations which
occur in particular when classifying vectorial Boolean functions with
good cryptographic properties.

\begin{problem}[EA-partitioning]\label{problem:ea-testing}
  Let $\{ F_{i} \}_{0 \leq i < \ell}$ be a set of $\ell$
  functions from $\ftwo^{n}$ into
  $\ftwo^{m}$. Partition this set in such a way
  that any two functions in the same subset are EA-equivalent and no two functions belonging to distinct subsets are EA-equivalent.
\end{problem}

Of course, EA-partitioning can be solved by applying an algorithm for
EA-recovery to each pair \((F_i, F_j)\), \(i < j\).  But, since we
mainly focus on situations where the set of functions to be
partitioned is large (for instance, around 20,000), examining all pairs of
functions would be too computationally expensive. A situation where EA-partitioning of such a large set of functions is
needed is related to one of the most prominent problems related to
optimal cryptographic functions: the search for so-called {\em Almost
  Perfect Nonlinear (APN)} functions, which are the functions offering
the best possible protection against differential
cryptanalysis~\cite{C:NybKnu92} (see e.g.~\cite{add:BloNyb15} for a
survey on APN functions). Most notably, the existence of APN bijective
mappings in even dimensions~\(n > 6\) is a long-standing open question
known as the {\em Big APN problem}~\cite{add:BDMW10}.  The only known
APN permutation in even dimension was exhibited in~2009 by Dillon and
his co-authors~\cite{add:BDMW10} in dimension~\(6\). It was derived
from an APN quadratic non-bijective function by a more general notion
of equivalence called {\em
  CCZ-equivalence}~\cite{add:CarChaZin98}. CCZ-equivalence expresses
the fact that the graphs of the functions, i.e., the sets of the form
\(\{(x,F(x)) : x \in \F_2^n\}\), are images of each other by an affine
transformation. This notion is strictly more general than
EA-equivalence~\cite{add:BudCar10} but the two notions coincide for
some particular cases including quadratic APN
functions~\cite{add:Yoshiara11}. Since Dillon {\em et al.}'s seminal
work, a standard strategy, followed by many
authors~\cite{add:YuWanLi14,add:CanDuvPer17,add:BeiLea20,add:WTG13},
to search for APN permutations in even dimension, consists in
searching for quadratic APN non-bijective functions and in exploring
their CCZ-equivalence class. To this end, it is important that the
explored CCZ-equivalence classes be distinct. This equivalently means
that they are generated from quadratic APN functions which are not
EA-equivalent. Following this approach, around \(20,000\) of quadratic
APN functions have been generated in~\cite{add:YuWanLi14} and
\cite{add:BeiLea20}, and it is necessary to check whether some of them
belong to the same EA-equivalence class (or, equivalently, the same
CCZ-class). The search for quadratic APN functions, which is motivated
by the {\em Big APN problem} but may be of independent interest, is
then the main relevant use-case for which we propose an efficient
EA-partitioning procedure, which amounts to partitioning the functions
into EA-classes based on their {\em ortho-derivatives}.

\subsection{State of the Art: Affine-Equivalence Recovery in some
  Specific Cases}

A generic procedure for testing CCZ-, EA- and affine equivalence
and, if it exists, for recovering the corresponding transformation, is
derived from coding theory. Actually, the main cryptographic properties
(e.g. the APN property) can be interpreted as conditions on some
binary linear codes, as first shown in~\cite{add:CarChaZin98}. To this
end, any function \(F\) from \(\F_2^n\) to \(\F_2^m\) is associated to the linear binary code \(\mathcal{C}_F\) of length \(2^n\) defined by the following
\((n+m+1) \times 2^n\) generator matrix
\[G_F = \left(\begin{array}{ccccc}
  1 & \ldots & 1 & \ldots & 1 \\
  0 & \ldots & x_i & \ldots & x_{2^n-1}\\
  F(0) & \ldots & F(x_i) & \ldots & F(x_{2^n-1})
\end{array}\right),\]
where \(\{0, x_1, \ldots, x_{2^n-1}\}\) denotes the set of all
elements in \(\F_2^n\) and each entry in the matrix is viewed as a
binary column-vector. As observed in~\cite{add:EdePot09b,add:BDKM09},
CCZ-equivalence then coincides with the usual notion of equivalence
between two linear binary
codes~\cite[Page~39]{MacSlo77}. EA-equivalence and affine equivalence
also correspond to code equivalence, but for codes defined by a
slightly different generator matrix (see Section~7
in~\cite{add:EdePot09b}).  A general technique for testing and
recovering CCZ- and EA-equivalence then consists in applying an
algorithm for testing the equivalence between two linear codes,
e.g.~\cite{add:Leon95}, as proposed
in~\cite{add:EdePot09b,add:Calderini20}. This algorithm is generic in
the sense that it applies to any linear binary code, not only to codes
associated to vectorial Boolean functions. But the downside is
that it is expensive in terms of both time and memory.
For instance, its implementation in SAGE~\cite{add:sage} requires at
least 40 seconds to check the CCZ-equivalence of two 8-bit functions, which
is often prohibitively expensive since this verification must be done
for many pairs of functions. {For 9-bit functions, it may crash
  from lack of memory, though the machine on which we performed this
  test has 32GB of RAM.} While Magma~\cite{add:Magma} might be faster,
  it is not as easily available (Magma is not free). More importantly,
  as noted in~\cite[p.~20]{add:nikolay}, the routine verifying
  CCZ-equivalence in Magma cannot be run for $n=12$ (even on a
  supercomputer with 500~GB of RAM). Even for $n \in \{9,10\}$, it
  sometimes requires several hours and may report wrong results.

However, since EA-recovery is a specific instance of the code
equivalence problem, it might be possible a more efficient algorithm if
we focus on codes of the form~\(\mathcal{C}_F\).  But very few results
are known even when \(F\) belongs to some particular families of
functions. Instead, algorithms for solving several particular cases,
when \(A\), \(B\) or \(C\) have a specific form, have been proposed
based on other techniques, for instance in~\cite{add:BudKaz12}
and~\cite{add:OzbSinYay14}.  The case $C = 0$ (corresponding to {\em
  affine equivalence}) has been solved when $F$ and $G$ are
permutations in the sense that we have algorithms capable of finding
$A$ and $B$ in this context~\cite{EC:BDBP03}.

\subsubsection{Guess-and-Determine}

The first algorithm for affine equivalence recovery was proposed
in~\cite{EC:BDBP03}. It is based on a subroutine which returns the
``linear representative'' of a permutation. Given a permutation $F$,
it returns the two linear permutations $L_{0}$ and $L_{1}$ such that, among
all the possible choices for $L_0, L_1$, the linear permutation
$L_{1} \circ F \circ L_{0}$ is the smallest with respect to the lexicographic
order. This algorithm is based on a guess-and-determine approach. Its
authors estimated its time complexity to be $\bigoh{n^{3}2^{n}}$ if
$F(0) \neq 0$ and $\bigoh{n^{3}2^{2n}}$ otherwise.

We have implemented this algorithm and, in practice, the running time
can be worse than this. Indeed, the complexity analysis assumes that a
contradiction in the guess-and-determine procedure will occur after a
small number of guesses.
This is usually true but, in some cases, it may happen that
all values for a variable need to
  be tested. In this case, the claimed complexity is multiplied by a
  factor \(2^n\).

  Using this algorithm, it is easy to recover $A$ and $B$ when they
  are linear. However, when they are affine, we also need to perform
  an exhaustive search for the constants \(A(0)\) and \(B(0)\). In
  this case, we generate two lists containing the linear
  representatives of $x \mapsto F(x \oplus a)$ and
  $x \mapsto b \oplus G(x)$. We then check if there is a match in
  these lists, i.e., an entry that belongs to both lists. Both the
  time and memory complexities in this case are multiplied by
  $2^{n}$. The overall time is then $\bigoh{n^{3}2^{2n}}$ (if we assume
  that the complexity estimation of the authors of~\cite{EC:BDBP03} is
  correct).

This method works for all permutations. But its downsides are that its
complexity is underestimated in some cases and that it
does not apply to non-bijective mappings.

\subsubsection{Rank Table}\label{sec:dinur}

In a more recent paper~\cite{EC:Dinur18}, Dinur proposed a completely
different approach based on so-called ``rank tables''. Paraphrasing
the introduction of said paper, the main idea of the algorithm is to
compute the rank tables of both $F$ and $G$ and then use these tables
to recover the affine transformation $B$, assuming that
$G = A \circ F \circ B$.  The rank tables of \(F\) and \(G\) are
obtained as follows.  We derive from $F$ (resp. from \(G\)) several
functions, each one defined by restricting its \(2^n\) inputs to an
affine subspace of dimension $n-1$. Each function derived in this way
has an associated rank, corresponding to the rank of the matrix formed
by the vectors representing the algebraic normal forms of its
coordinates (see Section~\ref{sec:preliminaries} for the definitions of
the notions). We then assign to each possible $(n-1)$-dimensional
subspace a corresponding rank. As there are $(2^{n+1}-2)$ possible
affine subspaces, we obtain $(2^{n+1}-2)$ rank values for $F$
(resp. for \(G\)).  These values are collected in the rank table of
$F$ (resp. \(G\)), where a rank table entry $r$ stores the set of all
affine subspaces to which a rank of~\(r\) has been assigned. We then
look for matches in these two rank tables.

This approach is faster than the algorithm of Biryukov {\em et al.},
as the computational time for affine equivalence recovery is
$\bigoh{n^{3}2^{n}}$, i.e. it is $2^{n}$ times faster. Unlike the
latter, the rank table-based approach works even if the functions are
not bijective; but it does require that their algebraic degree is high
enough, i.e. at least $n-2$~\cite{EC:Dinur18}. We have used an
implementation of this algorithm by its author and we have confirmed
that it could very efficiently handle non-bijective functions of
degree $n-1$.  However, for functions of degree $n-2$, the rank tables
of the test functions turn out to contain a single value. In such a
situation the algorithm is not faster than an exhaustive search, and
is therefore not relevant.

\subsection{Our Results}
While the two previously mentioned algorithms are dedicated to {\em
  affine equivalence} recovery, i.e, to the case \(C=0\), we present
here the first efficient algorithm for {\em EA-equivalence} recovery
in any dimension, 
when the involved functions \(F\) and
\(G : \F_2^n \rightarrow \F_2^m\) are quadratic. Our algorithm works both on
permutations and non-bijective functions. We prove that its
complexity depends on the differential spectrum of the
function and is estimated to be
$\bigoh{\max(n,m)^\omega 2^n + R^s (m^2+n^2)^{\omega}}$, where
$\omega \simeq 2.37$ denotes the complexity exponent of matrix
multiplication and $R$ is the number of vectors
$v \in \F_2^n \setminus \{0\}$ at which the rank of the Jacobian
matrix is the smallest possible. The last parameter $s$ is a number of
guesses which, when $m=n$, can be chosen to be equal to $3$.  Hence,
the estimated complexity is $\bigoh{n^\omega 2^n + R^3 n^{2\omega}}$
and it turns out that, for vectorial Boolean functions chosen at
random, the quantity $R$ is frequently very small.  On the other hand,
the most difficult case corresponds to APN functions where the
complexity is $\bigoh{2^{2n} (m^2+n^2)^{\omega}}$.

The second part of the paper details several tools for solving the
{\em EA-partitioning problem}. Most notably, we propose some new and
very efficient EA-invariants for quadratic APN functions, which is one
of the most important instances of this problem. This technique is
then used to partition the CCZ-classes of all the \(6\)-bit quadratic
APN functions into EA-classes. Also, by applying this method to
\(8\)-bit quadratic APN functions, we show that it is by far the most
efficient one for solving Problem~\ref{problem:ea-testing} in the case
of quadratic APN functions.

It is worth noting that, as detailed in Table~\ref{tab:recap}, only
some problems related to EA-equivalence have been solved. Several
simplified cases have been solved
in~\cite{add:BudKaz12,add:OzbSinYay14}, where some of the involved
affine functions correspond to the identity or to the addition of a constant. 
In a
  very recent and concurrent work~\cite{add:Kaleyski20} (independent
  of ours) another EA-recovery algorithm is presented by Kaleyski
  which is 
  based on a new EA-invariant called $\Sigma^{k}$-multiplicities. This
  algorithm is efficient and is not restricted to quadratic
  functions but it applies to even dimensions only. Indeed, it has
  been observed that the underlying invariant does not provide any
  useful information in odd dimensions. Its complexity, while low
  enough in practice, is hard to estimate due to its reliance on backtracking. However, finding a general
algorithm for EA-recovery, faster than generic techniques for
recovering code equivalence, for functions of degree strictly
greater than two and for any dimension, remains an open
problem.

\if\IEEEversion1
\begin{table*}[h!tb]
\else
\begin{table}[h!tb]
\fi
  \centering
  \renewcommand\arraystretch{1.3}
  \setlength\tabcolsep{8pt}
  \begin{tabular}{lcc}
    \toprule
    Condition & Complexity & Reference \\
    \midrule
    $C=0$, $m=n$, \(F\) and \(G\) bijective & $\bigoh{n^{3}2^{2n}}$ & \cite{EC:BDBP03} \\
    $C=0$, $\deg(F) \geq n-1$ & $\bigoh{n^{3}2^{n}}$ & \cite{EC:Dinur18} \\
    $A(x)=x \oplus a, B(x) = x \oplus b$ & $\bigoh{n2^{n}}$ & \cite{add:BudKaz12} \\
    $A(x) = x, B(0) = 0$ & $\bigoh{2^{n}}$ & \cite{add:OzbSinYay14} \\
    $B(x) = x \oplus b$ & $\bigoh{m2^{3n}}$ & \cite{add:BudKaz12} \\
    $\deg(F)=2$ & $\bigoh{n^{2\omega}2^{2n}}$ & Section~\ref{sec:jacobian-solving} \\
    $n$ even & -- & \cite{add:Kaleyski21} \\
    \bottomrule
  \end{tabular}
  \caption{\label{tab:recap}Algorithms solving the EA-recovery of $F :
    \ftwo^{n} \to \ftwo^{m}$ and $G=A \circ F \circ B + C$. See
    Section~\ref{sec:invariants-list} for an overview of EA-partitioning.}
\if\IEEEversion1
\end{table*}
\else
\end{table}
\fi

\paragraph{Organization of the Paper.}

We first recall the basic concepts and definitions needed in
Section~\ref{sec:preliminaries}.

The rest of the paper successively presents our results that can
efficiently tackle both Problems~\ref{problem:ea-recovery}
and~\ref{problem:ea-testing} in the case of quadratic
functions. First, we show how to reduce EA-recovery to the resolution
of a linear system using the Jacobian matrix. This approach is
described in Section~\ref{sec:jacobian}. 

Then, Section~\ref{sec:invariants} first describes a general approach based
on class invariants which can solve EA-partitioning
(Problem~\ref{problem:ea-testing}), and
lists all the CCZ- and EA-class invariants we are
aware of from the literature. Then, in the case of quadratic APN
functions, we introduce a new invariant based on
\emph{ortho-derivatives}. It is very discriminating, and can
efficiently prove that more than $20,000$ distinct quadratic APN
functions of \(8\)~variables fall into different CCZ-classes in only a
few minutes on a regular desktop computer.

Our optimized implementations of all these invariants are available
in the Sage package \texttt{sboxU}.\footnote{\texttt{sboxU} is
  available for download at
  \url{https://github.com/lpp-crypto/sboxU}.}

 \section{Preliminaries and Definitions}
\label{sec:preliminaries}

We consider vectorial Boolean functions, that is functions from the
vector space $\F_2^n$ to the vector space $\F_2^m$ for some non-zero
$m$ and $n$. When \(m=1\),
such functions are called Boolean functions. Any vectorial Boolean
function can be represented in different ways. For instance, it can be
seen as a sequence of \(m\)~Boolean
functions from \(\F_2^n\)
to \(\F_2\),
called its {\em coordinates}. Each of these \(m\)~coordinates
can be uniquely written as a multivariate polynomial in
\(\F_2[x_1, \ldots, x_n]/\left(x_1^2+x_1, \ldots, x_n^2 +
  x_n\right)\),
called its {\em Algebraic Normal Form (ANF)}.  A vectorial Boolean
function can also be represented by its {\em truth table} (aka look-up
table), which is the array composed of all its output values.  The
following notions will be extensively used throughout the paper.

\paragraph{Differential Properties.}
The resilience of a function to \emph{differential
  attacks}~\cite{C:BihSha90} is expressed by simple properties of its
derivatives.  
\begin{definition}[Derivative]
  Let $F$ be a function from $\F_2^n$ into $\F_2^m$.  The
  \emph{derivative of $F$ with respect to $a \in \F_2^n$} is the
  function from $\F_2^n$ into $\F_2^m$ defined by
 \[\Delta_a F : x \in \F_2^n \mapsto F(x + a) + F(x).\]
\end{definition}
In practice, the properties of the derivatives are analyzed through
the multiset introduced below which corresponds to the entries of its
{\em difference distribution table (DDT)}.

\begin{definition}[DDT, differential spectrum]
  Let $F$ be a function from $\F_2^n$ into $\F_2^m$.  The DDT of \(F\)
  is the \(2^n \times 2^m\)~array consisting of all elements
  \if\IEEEversion1
  \[\delta_F(a,b) := \# \{ x\in \F_2^n~: F(x+a)+F(x)=b\},\]
  for $(a,b) \in \F_2^n \times \F_2^m$. 
  \else
  \[\delta_F(a,b) := \# \{ x\in \F_2^n~: F(x+a)+F(x)=b\}, \mbox{ for } (a,b)
    \in \F_2^n \times \F_2^m.\]
  \fi
  The {\em differential uniformity} of \(F\)~\cite{EC:Nyberg93} is defined as
  \[\delta(F) := \max_{a \in \F_2^n \setminus\{0\}, b \in \F_2^m} \delta_F(a,b),\]
  and the {\em differential spectrum of \(F\)} is the multiset
  \[[\delta_F(a,b), a \in \F_2^n, a \neq 0, b \in \F_2^m].\]
\end{definition}
Obviously, \(\delta(F) \geq 2^{n-m}\) and the functions for
which equality holds are called {\em  Perfect Nonlinear} or {\em bent}. Such functions exist only when \(n\) is even and \(m \leq n/2\)~\cite{EC:Nyberg91}.
When \(m\geq n\), the differential uniformity satisfies \(\delta(F) \geq 2\), and the functions for
which equality holds are called {\em Almost Perfect Nonlinear (APN)}
functions.

\paragraph{Walsh Transform.}
Similarly, the resistance of a function to linear attacks~\cite{EC:Matsui93} is evaluated through its {\em Linear approximation table (LAT)}, whose entries are given by the Walsh transform.
\begin{definition}[Walsh transform]
Let $F$ be a function from $\F_2^n$ into $\F_2^m$. Its Walsh transform at \((a,b) \in \F_2^n \times \F_2^m\) is the integer defined by
\[\walsh{F}(a, b) := \sum_{x \in \ftwo^{n}}(-1)^{\scalarprod{a}{x} + \scalarprod{b}{F(x)}}\]
where $\scalarprod{y}{z}$ denotes the
canonical inner product on $\F_2^n$, {\em i.e.}
$\scalarprod{y}{z} := \sum_{i=1}^n y_i z_i$ where \(y=(y_1, \ldots, y_n)\) and \(z=(z_1, \ldots, z_n)\).
The {\em Walsh spectrum} of \(F\) is then the multiset
\[\left[\walsh{F}(a,b) : a \in \F_2^n, b \in \F_2^m\right].\]
\end{definition}

\paragraph{Degree.}

The degree of a vectorial Boolean function is defined as follows.
\begin{definition}[Degree]
  Let $F$ be a function from $\F_2^n$ into $\F_2^m$. The {\em degree}
  of \(F\) is the maximal degree of the algebraic normal forms of its
  coordinates. 
\end{definition}
The functions of degree less than or equal to~\(1\) are said to be {\em
  affine}. Quadratic functions, {\em i.e.} functions of degree~\(2\),
will also play an important role in this paper.

\paragraph{CCZ-Equivalence.}

While this paper focuses on EA-equivalence, there exists a more
general notion of equivalence between vectorial Boolean functions
defined by Carlet, Charpin and Zinoviev~\cite{add:CarChaZin98} and
called {\em CCZ-equivalence}. This notion will be widely used in
Section~\ref{sec:invariants}.
\begin{definition}
  Two functions $F: \ftwo^{n} \to \ftwo^{m}$ and
  $G: \ftwo^{n} \to \ftwo^{m}$ are CCZ-equivalent if there exists an
  affine permutation $\mathcal{A}$ of \(\ftwo^{n} \times \ftwo^{m}\)
  such that
  \[\mathcal{A}\left(\{(x, F(x)): x \in \F_2^n\}\right) = \{(x, G(x)):
    x \in \F_2^n\}.\]
\end{definition}
Two functions which are EA-equivalent are also
CCZ-equivalent  (see~\cite[Page~29]{add:Carlet21} for a proof), but the converse does not hold~\cite{add:BudCarPot06}.

In general, given a function $F : \ftwo^{n} \to \ftwo^{m}$ and an
affine permutation $\mathcal{A}$ of $\ftwo^{n} \times \ftwo^{m}$,
\begin{equation*}
  \mathcal{A}\big(\left\{ (x, F(x)): x \in \ftwo^{n} \right\}\big) 
\end{equation*}
does not correspond to the graph of a function, i.e. to a set of the
form \(\{(x, G(x)): x \in \F_2^n\}\) for some function \(G\) from
\(\F_2^n\) to \(\F_2^m\).  Indeed,  for $G$ to be
well-defined, it is necessary that the first \(n\)
  coordinates of \(x \mapsto \mathcal{A}(x, F (x))\) permute  \(\F_2^n\). As a consequence, only a few permutations $\mathcal{A}$
yield valid functions $G$. The following definition captures this
intuition.
\begin{definition}[Admissible affine permutations]
  Let $F$ be a function from $\ftwo^{n}$ to $\ftwo^{m}$. We say that
  the affine permutation $\mathcal{A}$ of $\ftwo^{n} \times \ftwo^{m}$
  is \emph{admissible for $F$} if 
  \[\mathcal{A}\left(\{(x, F(x)), x \in \F_2^n\}\right)\]
is a graph of a function.
\end{definition}
 \section{Recovering EA-equivalence for Quadratic Functions}
\label{sec:jacobian}

\subsection{The Jacobian Matrix}

\begin{notation}\label{nota:can_basis}
In the sequel, the canonical basis of $\F_2^n$ is denoted by
$(e_1, \dots, e_n)$.
\end{notation}

\begin{definition}\label{def:jacobian}
Let \(F\) be a function from \(\F_2^n\) into \(\F_2^m\).
The \emph{Jacobian} of $F$ at $x \in \F_2^n$ is the
matrix with polynomial entries defined by
\begin{equation}
  \if\IEEEversion1
  \!\!\!\!\jacobian{F}{x}
  :=
    \begin{pmatrix}
      \Delta_{e_1}F_{1}(x) & \!\! \cdots \!\! &
                                             \Delta_{e_n}F_{1}(x) \\
      \vdots & \!\! \ddots \!\! & \vdots \\
     \Delta_{e_1}F_{m}(x) & \!\! \cdots \!\! &
                                             \Delta_{e_n}F_{m}(x)  \\
    \end{pmatrix}\cdot
  \else
  \jacobian{F}{x}
  ~:=~
    \begin{pmatrix}
      \Delta_{e_1}F_{1}(x) & \cdots &
                                             \Delta_{e_n}F_{1}(x) \\
      \vdots & \ddots & \vdots \\
     \Delta_{e_1}F_{m}(x) & \cdots &
                                             \Delta_{e_n}F_{m}(x)  \\
    \end{pmatrix}\cdot
  \fi
  \end{equation}
  On the other hand, given an $m$--tuple of polynomials
  $P = (P_1, \dots, P_m) \in \F_2[X_1, \dots, X_n]^m$, we define
  the Jacobian matrix of $P$ as
  \if\IEEEversion1
  the matrix
  \[
    \JacPoly P(x) ~:=~
    \begin{pmatrix} 
      \partialDiff{P_{1}}{x_{1}} & \cdots &
      \partialDiff{P_{1}}{x_{n}} \\
      \vdots & & \vdots \\
      \partialDiff{P_{m}}{x_{1}} & \cdots &
      \partialDiff{P_{m}}{x_{n}} 
    \end{pmatrix}
  \]
  in $\F_2[X_1, \dots, X_n]^{m\times n}$,
  \else
    \[
    \JacPoly P(x) ~:=~
    \begin{pmatrix} 
      \partialDiff{P_{1}}{x_{1}} & \cdots &
      \partialDiff{P_{1}}{x_{n}} \\
      \vdots & & \vdots \\
      \partialDiff{P_{m}}{x_{1}} & \cdots &
      \partialDiff{P_{m}}{x_{n}} 
    \end{pmatrix}
    \in \F_2[X_1, \dots, X_n]^{m\times n},
  \]
  \fi
    where \(\partialDiff{P_{i}}{x_{j}}\) denotes the partial derivative of \(P_i\) with respect to~\(x_j\).
\end{definition}

\begin{remark}\label{rem:Jacobian_in_diff_calculus}
  The two notions are strongly related to each other. In particular,
  for a vectorial Boolean function $F : \F_2^n \rightarrow \F_2^n$,
  denote by $P^{\text{ANF}}_F \in \F_2[X_1, \dots, X_n]^m$ the
  polynomial representation of $F$ {\bf in algebraic normal form},
  then
  \[
    \jacobian{F}{x} = \JacPoly P^{\text{ANF}}_F(x).
  \]
  This equality can be easily checked on monomials and then extended by
  linearity.

  Note however the importance of being in algebraic normal form: for instance
  in one variable, if $P(x) = x^2$, then $\frac{\partial P}{\partial x} = 0$,
  while the algebraic normal form of~\(P\) is $x$ whose derivative is $1$.
\end{remark}

\subsection{The Jacobian Matrices of EA-equivalent Functions}

Assume that $G = A \circ F \circ B + C$, for some affine permutations
$A$ and $B$ and an affine function $C$ defined as
\if\IEEEversion1
  \begin{align*}
       \forall x \in \F_2^n,\quad
 A(x) &= A_0  x + a,\\
  B(x) &= B_0  x + b,\\
   \mbox{and}\quad C(x) &= C_0 x + c,
  \end{align*}
\else
\[
  \forall x \in \F_2^n,\quad
  A(x) = A_0  x + a,\quad
  B(x) = B_0  x + b,\quad \mbox{and}
  \quad
  C(x) = C_0 x + c,
\]
\fi
where $A_0, B_0$ are non-singular matrices in $\F_2^{n\times n}$ and
$\F_2^{m \times m}$ respectively, \(C_0\) is a matrix in $\F_2^{n\times m}$, and $a, c \in \F_2^m$ and
$b \in \F_2^n$.
Note that, after replacing $a$ by $a+c$, one can assume that $c=0$ and
hence that $C$ is linear. We always do so in the sequel.

Denote by $P_F$ some polynomial representation of
$F$ (which is not necessarily its ANF). Then consider the polynomial
representation $P_G$ of $G$ defined by
\[
  P_G := A \circ P_F \circ B + C.
\]
Then, considering Jacobian matrices, one can apply the well--known
 {\em chain rule formula for functions of several variables},
{\em i.e.} the formula for the Jacobian of compositions of functions,
namely:
\begin{equation}
  \label{eq:jacobian-composition}
  \JacPoly P_G(x) = A_0 \cdot \JacPoly P_F (B(x)) \cdot B_0 + C_0. 
\end{equation}
Indeed, the Jacobians of $A, B$ and $C$ are $A_0, B_0$ and
$C_0$, respectively. 
Unfortunately, this chain rule formula does not extend to Jacobians
of Boolean functions because, as already observed in
Remark~\ref{rem:Jacobian_in_diff_calculus}, the operations of
derivation and of reduction to the algebraic normal form do not
commute. To clarify this issue, let us consider an elementary example.

\begin{example}
  Let $P_F(x_1, x_2) = x_1x_2$ be a polynomial representing a
  function $F : \F_2^2 \rightarrow \F_2$. Consider the linear map
  $B : (x_1, x_2) \mapsto (x_1 + x_2, x_2)$.  Set
  $P_G := P_F \circ B = x_1x_2 + x_2^2$. This polynomial represents a function
  $G$ whose algebraic normal form is
  $P^{\text{ANF}}_G = x_1 x_2 + x_2$. Now, the Jacobian matrices of $F$ and $G$ are
  \if\IEEEversion1
  \begin{align*}
    \jacobian{F}{x} &= (x_2 \ \ x_1)\\
    \mbox{and}
    \quad \jacobian{G}{x} &= (x_2\ \ x_1 + 1)
  \end{align*}
  \else
  \[
    \jacobian{F}{x} = (x_2 \ \ x_1) \quad \mbox{and}
    \quad \jacobian{G}{x} = (x_2\ \ x_1 + 1)
  \]
  \fi
  and
  \if\IEEEversion1
  \begin{align*}
    \jacobian{F}{B(x)}
    \begin{pmatrix}
      1 & 1 \\ 0 & 1
    \end{pmatrix}
    & = (x_2\ \ x_1 + x_2) \begin{pmatrix}
      1 & 1 \\ 0 & 1
    \end{pmatrix}\\
    & = (x_2\ x_1),
  \end{align*}
  \else
  \[
    \jacobian{F}{B(x)}
    \begin{pmatrix}
      1 & 1 \\ 0 & 1
    \end{pmatrix}
    = (x_2\ \ x_1 + x_2) \begin{pmatrix}
      1 & 1 \\ 0 & 1
    \end{pmatrix}
    = (x_2\ x_1),
  \]
  \fi
  which differs from $\jacobian{G}{x}$.  On the other hand, if we
  consider polynomials instead of Boolean functions, we have
  \[
    \JacPoly P_F(B(x)) \begin{pmatrix}
      1 & 1 \\ 0 & 1
    \end{pmatrix} = (x_2\ x_1) = \JacPoly P_G (x).
  \]
 \end{example}

 While we observed that the chain rule formula in several variables is
 false in general for Boolean functions, we prove in the sequel that,
 in the context of quadratic functions, it is possible to get a very
 similar formula using the so-called {\em linear part} of the
 Jacobian. This will be the central idea of our algorithm.

\subsection{The Jacobian Matrix of a Quadratic Function}

From now on, we assume that \(F\) is quadratic, i.e. its algebraic
normal form has degree $2$. In this case, the entries of the
associated Jacobian matrix, in the sense of
Definition~\ref{def:jacobian}, are polynomials of degree $1$, and we
will focus on their homogeneous parts.

\begin{definition}
  Let $F : \F_2^n \rightarrow \F_2^m$ be a quadratic function.  We
  denote by \(\Jlin{F}{x}\) the linear part of $\jacobian{F}{x}$,
  i.e. the matrix whose entries are the homogeneous parts of degree $1$
  of the entries of \(\jacobian{F}{x}\):
  \[
    \forall x \in \F_2^n,\if\IEEEversion1
    \ \ 
    \else
    \quad
    \fi
    \jacobian{F}{x} = \Jlin{F}{x} +
    \jacobian{F}{0}.
  \]
Equivalently,
  \(\Jlin{F}{x} = {(J_{i,j}(x))}_{i,j}\) with
  \if\IEEEversion1
  for any $x\in \F_2^n$ and $1 \leq i \leq m, 1 \leq j \leq n$,
\begin{equation}\label{eq:J}
J_{i,j}(x) = \Delta_{e_j} F_i(x) + \Delta_{e_j} F_i(0),
\end{equation}
  \else
\begin{equation}\label{eq:J}
J_{i,j}(x) = \Delta_{e_j} F_i(x) + \Delta_{e_j} F_i(0), \; \forall x \in \F_2^n, 1 \leq i \leq m, 1 \leq j \leq n,
\end{equation}
\fi
where $(e_1, \dots, e_n)$ denotes the canonical basis of $\F_2^n$
(Notation~\ref{nota:can_basis}).
\end{definition}

It is worth noting that the linear part of the Jacobian of a
quadratic function corresponds to the coefficients of the quadratic
monomials in the algebraic normal forms of the coordinates of \(F\).
Equivalently, $\Jlin{F}{x}$ is the Jacobian of the degree-$2$
homogeneous part of the algebraic normal form of $F$ as explained by the
following statement.

\begin{proposition}Let $F : \F_2^n \rightarrow \F_2^m$ be a quadratic function.
Let
\[F_i(x_1, \ldots, x_n) = \sum_{k < \ell} Q_{k,\ell}^i x_k x_\ell +
  \sum_{k=1}^n c_k^i x_k + \varepsilon^i\] be the algebraic normal
  form of the \(i\)-th coordinate of \(F\), \(1 \leq i \leq m\), where
  all coefficients \(Q_{k,\ell}^i, c_k^i, \varepsilon^i\) lie in
  \(\F_2\) and \(Q_{k,\ell}^i=0\) when \(k \geq \ell\).  Then, the
  entries $J_{i,j}(x)$ of \(\Jlin{F}{x}\) are
  \if\IEEEversion1
  \[J_{i,j}(x) = \sum_{k=1}^n (Q_{k,j}^i + Q_{j,k}^i)x_k,\]
  for any $1 \leq i
  \leq m, 1 \leq j \leq n$.
\else
  \[J_{i,j}(x) = \sum_{k=1}^n (Q_{k,j}^i + Q_{j,k}^i)x_k, \; 1 \leq i
  \leq m, 1 \leq j \leq n.\]
\fi
\end{proposition}
\begin{proof}
  For any \(i\) and \(j\), \(1 \leq i \leq m\) and
  \(1 \leq j \leq n\), we have
\[\Delta_{e_j}F_i(x) = \sum_{k<j} Q_{k,j}^i x_k + \sum_{k>j} Q_{j,k}^i
  x_k + c_j^i.\]
We then deduce from~\eqref{eq:J} that
\if\IEEEversion1
\begin{align*}
  J_{i,j}(x) & = \Delta_{e_j} F_i(x) + \Delta_{e_j} F_i(0)\\
  & = \sum_{k=1}^n (Q_{k,j}^i + Q_{j,k}^i) x_k.
\end{align*}
\else
\[J_{i,j}(x) = \Delta_{e_j} F_i(x) + \Delta_{e_j} F_i(0) =
  \sum_{k=1}^n (Q_{k,j}^i + Q_{j,k}^i) x_k.\]
\fi
\end{proof}

The linear part of the Jacobian of a quadratic function is a useful
mathematical object since the values of all derivatives of the
function can be derived from this matrix, as shown in the following
proposition.

\begin{proposition}\label{prop:derivative}
  Let $F : \F_2^n \rightarrow \F_2^m$ be a quadratic function and let
  \[
    \Jlin{F}{x} = (J_{i,j}(x))_{1 \leq i \leq m, 1 \leq j
         \leq n}
  \]
  be the linear part of its Jacobian. Then, for any
  \(a = (a_1, \ldots, a_n) \in \F_2^n\), any \(i \in \{1, \dots, m\}\)
  and any $x \in \F_2^n$, we have:
  \[\Delta_aF_i(x) + \Delta_aF_i(0) = \sum_{j=1}^n a_j J_{i,j}(x).\]
  Hence, for any $a\in \F_2^n$ and any $x \in \F_2^n$, we have:
  \if\IEEEversion1
  \begin{align*}
    \Delta_a F(x) + \Delta_a F(0) &=
                                    \Delta_x \Delta_a F(0)\\
    &= \Jlin{F}{x} \cdot a\\ &= \Jlin{F}{a} \cdot x.
  \end{align*}
  \else
  \[
    \Delta_a F(x) + \Delta_a F(0) =  \Delta_x \Delta_a F(0) = \Jlin{F}{x} \cdot a = \Jlin{F}{a} \cdot x.
  \]
  \fi
\end{proposition}
\begin{proof}
Let us first observe that, for any \(a,x \in \F_2^n\), 
\if\IEEEversion1
\begin{align*}
  \Delta_a F(x) &+ \Delta_a F(0)\\ &= F(x+a)+ F(x)+F(a)+F(0)\\
  &= \Delta_x \Delta_a F(0)
\end{align*}
\else
\[\Delta_a F(x) + \Delta_a F(0) = F(x+a)+ F(x)+F(a)+F(0) = \Delta_x \Delta_a F(0)\]
\fi
and hence
\[\Delta_a F(x) + \Delta_a F(0)=  \Delta_x F(a) + \Delta_x F(0)\;.\]
Let \(S = \{i : a_i=1, 1 \leq i \leq n\}\). Then,
\begin{eqnarray*}
\sum_{j=1}^n a_j J_{i,j}(x) & = & \sum_{j \in S} \left(\Delta_{e_j} F_i(x) + \Delta_{e_j} F_i(0)\right)\\
& = & \sum_{j \in S} \left(\Delta_{x} F_i(e_j) + \Delta_{x} F_i(0)\right)\\
& = &\Delta_{x} F_i\left(\sum_{j \in S}e_j\right) + \Delta_{x} F_i(0) \\
                            & = & \Delta_{x} F_i(a) + \Delta_{x} F_i(0) \\
  & = & \Delta_{a} F_i(x) + \Delta_{a} F_i(x).
\end{eqnarray*}
\end{proof}

The following corollary is another direct consequence of Proposition~\ref{prop:derivative}.
  \begin{corollary}\label{cor:a_in_kerJ(a)}
    Let $F : \F_2^n \rightarrow \F_2^m$ be  a quadratic function.
    Then, for any $a \in \F_2^n$, we have
    \[
      \Jlin{F}{a} \cdot a = 0.
    \]
  \end{corollary}

  Using Proposition~\ref{prop:derivative}, we can exhibit the relation
  between the linear parts of the Jacobians of two EA-equivalent
  quadratic functions. This relation is very close to the chain rule
  formula in differential calculus.  In addition it will be of
  particular interest for recovering the triple of functions
  \((A,B, C)\) such that \(G = A \circ F \circ B + C\) because it does
  not involve~\(C\).

\begin{proposition}\label{prop:chain_rule_for_jacqueline}
  Let \(F\) and \(G\) be two EA-equivalent quadratic functions from
  \(\F_2^n\) into \(\F_2^m\) with \(G = A \circ F \circ B + C\) for
  some affine permutations $A$ and $B$, and some affine function
  \(C\).  
  Then, \if\IEEEversion1 for any $x \in \F_2^n$, \fi
\begin{equation}\label{eq:main}
\if\IEEEversion0 \forall x \in \F_2^n, \quad \fi
      \Jlin{G}{x} =
    A_0\cdot \Jlin{F}{B(x)} \cdot B_0,
\end{equation}
where \(A_0\) and \(B_0\) denote the matrices corresponding to the
linear parts of \(A\) and \(B\).  
\end{proposition}
\begin{proof}
  Let \(A_0, B_0\) and \(C_0\) denote the matrices corresponding to
  the linear parts of \(A\),\(B\) and \(C\).  Then, for any \(a\) and \(x\) in \(\F_2^n\),
  \if\IEEEversion1
\begin{align*}
  \Delta_a &G(x)  + \Delta_a G(0)\\
                & =  \Delta_a (A\circ F \circ B) (x) + \Delta_a (A\circ F \circ B) (0)\\
& =  A_0 \left( \Delta_a (F \circ B) (x) + \Delta_a (F \circ B) (0)\right)\\
& =  A_0 \left( \Delta_{B_0(a)} F (B (x)) + \Delta_{B_0(a)} F (B (0))\right).
\end{align*}
  \else
\begin{eqnarray*}
\Delta_a G(x) + \Delta_a G(0) & = & \Delta_a (A\circ F \circ B) (x) + \Delta_a (A\circ F \circ B) (0)\\
& = & A_0 \left( \Delta_a (F \circ B) (x) + \Delta_a (F \circ B) (0)\right)\\
& = & A_0 \left( \Delta_{B_0(a)} F (B (x)) + \Delta_{B_0(a)} F (B (0))\right).
\end{eqnarray*}
\fi
It follows that Column~\(j\) of \(\Jlin{G}{x}\) is equal to
\if\IEEEversion1
\begin{align*}
  A_0 \big[ \Delta_{B_0(e_j)} &F (B (x)) + \Delta_{B_0(e_j)} F (B (0))\big]\\
  &= A_0 \cdot \Jlin{F}{B(x)} \cdot B_0(e_j),
  \end{align*}
\else
\[A_0 \left[ \Delta_{B_0(e_j)} F (B (x)) + \Delta_{B_0(e_j)} F (B (0))\right] = A_0 \cdot \Jlin{F}{B(x)} \cdot B_0(e_j),\]
\fi
where the equality is derived from Proposition~\ref{prop:derivative}. Therefore, Column~\(j\) of \(\Jlin{G}{x}\) is equal to Column~\(j\) of \(A_0 \cdot \Jlin{F}{B(x)} \cdot B_0\).
\end{proof}

It is worth noting that, when \(m=n\),
the linear part of the Jacobian, for quadratic functions defined by a
homogeneous polynomial\footnote{Obviously, any quadratic function is EA-equivalent to a function corresponding to a univariate homogeneous polynomial of degree~\(2\).} in \(\F_{2^n}[X]\),
is related to a special class of symmetric matrices over \(\F_{2^n}\)
called QAM and introduced in~\cite{add:YuWanLi14}. This notion of QAM
arises by exhibiting a one-to-one correspondence between the functions
over \(\F_2^n\)
associated to a quadratic homogeneous univariate polynomial (i.e.,
with quadratic terms only) and symmetric matrices over \(\F_{2^n}\)
with diagonal elements equal to zero~\cite{add:YuWanLi14}. This
correspondence is detailed in the following definition.

\begin{definition}[{\cite{add:YuWanLi14}}]\label{def:qam}
  Let \(\mathcal{B}=(\alpha_1, \alpha_2, \ldots, \alpha_n)\) be a
  basis of \(\F_{2^n}\) over \(\F_2\) and
  \(\varphi\) be the function from \(\F_2^n\) to \(\F_{2^n}\) defined by \(\varphi(x_1, \ldots, x_n) =\sum_{i=1}^n x_i
  \alpha_i\).  Let 
  \(P(X) = \sum_{k < \ell} q_{k,\ell} X^{2^{k-1} + 2^{\ell-1}}\) be
  a quadratic homogeneous polynomial in \(\F_{2^n}[X]\), and \(F: \F_2^n \rightarrow \F_2^n\) be the
  quadratic vectorial Boolean function corresponding to~\(P\), i.e., defined by
  \[F(x) = \varphi^{-1} \circ P \circ \varphi(x), \; \mbox{ for all } x \in
    \F_2^n.\] Then the {\em matrix associated with \(F\) with
    respect to \(\mathcal{B}\)} is
  \[H=M^T C_P M,\] where \(C_P\) is the \(n \times n\) symmetric matrix
  over \(\F_{2^n}\) defined by
  \((C_P)_{k,\ell}= (C_P)_{\ell,k} = q_{k,\ell}\) for all
  \(1 \leq k < \ell \leq n\), \((C_P)_{k,k}=0\), and \(M\) is the
  Moore matrix associated to \(\mathcal{B}\), i.e.,
  \(M_{i,j} = \alpha_j^{2^{i-1}}\).
  \end{definition}

  The advantage of this construction is that, when a function corresponds to a quadratic
  homogeneous univariate polynomial, being APN can be characterized by some
  algebraic properties of the associated matrix~\cite{add:YuWanLi14}.
  We now show that the matrix \(H=M^T C_P M\) is related to the linear
  part of the Jacobian of \(F\).
\begin{proposition}\label{prop:QAM}
  Let \(F\) be a function from \(\F_2^n\) to \(\F_2^n\)
  corresponding to a quadratic homogeneous univariate polynomial for
  some basis \(\mathcal{B}=(\alpha_1, \alpha_2, \ldots, \alpha_n)\),
  and let \(H\)
  be the matrix associated with \(F\)
  with respect to \(\mathcal{B}\)
  in the sense of Definition~\ref{def:qam}. Then, for all
  \(x\in \F_2^n\),
  \[(\alpha_1, \ldots, \alpha_n) \cdot \Jlin{F}{x} = (x_1, \ldots, x_n) \cdot H.\]
\end{proposition}
\begin{proof}
  Let \(P\)
  be the homogeneous univariate polynomial associated to \(F\)
  by \(P = \varphi \circ F \circ \varphi^{-1}\),
  where \(\varphi(x_1, \ldots, x_n) =\sum_{i=1}^n x_i \alpha_i\).
  Theorem~1 in~\cite{add:YuWanLi14} shows that, for any \(j\),
  \(1 \leq j \leq n\), and any \(X \in \F_{2^n}\),
\[P(X+\alpha_j) + P(X) + P(\alpha_j) = \varphi(X) \cdot H_j,\]
 where \(H_j\)
  denotes Column~\(j\) of~\(H\).
It follows that, for any \(x \in \F_2^n\),
\[(P\circ \varphi)(x+e_j) + P\circ \varphi (x) + P\circ \varphi (e_j) = x  H_j\]
or equivalently
\[\varphi\left(F(x+e_j)+F(x)+F(e_j)\right) = x  H_j,\]
  which means that
  \[
    \sum_{i=1}^n \alpha_i [\Jlin{F}{x}]_{i,j} =
    x H_j.
  \]
  \end{proof}
In other words, \(\Jlin{F}{x}\)
  is three-dimensional in nature: when \(m=n\),
  it is an \(n \times n\)
  matrix whose entries are \(n\)-variable
  linear functions, represented by \(n\)~binary
  coefficients. The previous proposition shows that the matrix \(H\)
  defined in~\cite{add:YuWanLi14} is another way to represent the same
  \(2^{3n}\)~binary
  values, with a 2-dimensional structure by using $\F_{2^n}$ as a
  coefficient field.  Characterizing the fact that \(F\)
  is APN from the properties of~\(H\)
  can then be reformulated in terms of Jacobians as we will prove in
  Section~\ref{section:diff}. However, our result applies also to functions from \(\F_2^n\)
  to \(\F_2^m\) with \(m \neq n\).
  More importantly, the representation in terms of Jacobians is more
  convenient for analyzing EA-equivalence.

\subsection{Solving the EA-equivalence Problem for Quadratic
  Functions}
\label{sec:jacobian-solving}

Our algorithm takes as input two quadratic functions
$F,G : \F_2^n \rightarrow \F_2^n$ and returns, if it exists, a triple
$(A, B, C)$ of affine functions such that $A, B$ are permutations and
$G = A \circ F \circ B + C$. Denote
\if\IEEEversion1
\begin{align*}
  A(x) & = A_0 x + a \\
  B(x) & = B_0x + b \\
  \text{and} \quad C(x) & = C_0x + c,
\end{align*}
\else
\[
  A(x) = A_0 x + a \quad B(x) = B_0x + b \quad \textrm{and} \quad C(x)
  = C_0x + c,
\]
\fi
where $A_0 \in \F_2^{m \times m}$, $B_0\in \F_2^{n\times n}$,
$C_0 \in \F_2^{m\times n}$, $a \in \F_2^m$, $b \in \F_2^n$
and $c\in \F_2^m$.

As already noted, replacing $a$ by $a+c$ one can assume that
$c = 0$ and hence that $C$ is linear. In addition, since the functions
are quadratic, one may also assume that $b = 0$. Indeed, it suffices to
observe that
\[
  \if\IEEEversion1
  \forall x\in \F_2^n, \ \
  \fi
  F(B_0x + b) = F(B_0 x) + \Delta_b F (B(x))
  \if\IEEEversion0
  , \mbox{ for all } x \in \F_2^n
  \fi
\]
and, since $F$ is quadratic, then $\Delta_b F (B(x))$ is affine and
its linear part equals \(\Delta_b F (B_0x) + \Delta_b F (0)\).
Therefore, replacing $C$ by
$x \mapsto C(x) + \Delta_b F (B_0 x) + \Delta_b F(0)$ and $a$ by
$a + \Delta_b F(0)$, one can assume that both $B$ and $C$ are linear.

In summary, our objective is to find, if it exists, a $4$--tuple
$(A_0, B_0, C_0, a)$ such that $A_0, B_0$ are non-singular, and
\[
  \if\IEEEversion1
  \forall x \in \F_2^n,\ \
  \fi
  G(x) = A_0 \cdot F(B_0 x) + C_0 x + a
  \if\IEEEversion0
  , \mbox{ for all } 
  x\in \F_2^n.
  \fi
\]
Our algorithm is built around Proposition~\ref{prop:chain_rule_for_jacqueline} which asserts
that,\if\IEEEversion1
  for any $x \in \F_2^n$,
  \fi
\begin{equation}\label{eq:chain_rule_Jacqueline}
  \Jlin{G}{x} = A_0 \cdot \Jlin{F}{B_0
    x}\cdot B_0
  \if\IEEEversion0
  , \mbox{ for all } 
  x\in \F_2^n.
  \fi
\end{equation}
This permits first to search for the pair $(A_0, B_0)$, and then, once
it is computed, to deduce the remainder of the $4$--tuple.
The search for this pair $(A_0, B_0)$ relies on two main ideas:

\begin{enumerate}[(i)]
\item\label{item:linear_system}
  If a pair $(v, w)$ is known to satisfy $B_0 v = w$, then the
  pair $(A_0^{-1}, B_0)$ is a solution of the affine system with
  unknowns $(X, Y) \in \F_2^{m\times m} \times \F_2^{n\times n}$:
  \begin{equation}\label{eq:main_linear_system}
    \if\IEEEversion1
    \!\!\!\!\!\!\!\!\!\!\!\!\!\!\!\!
    \fi
    \left\{
      \begin{array}{ccc}
        X \cdot \Jlin{G}{v} - \Jlin{F}{w} \cdot Y \if\IEEEversion1
        \!\!\!\!\fi &=& \if\IEEEversion1
        \!\!\fi 0\\
        Y \cdot v &=& \if\IEEEversion1
        \!\!\fi w.
      \end{array}
    \right.
  \end{equation}
\item\label{item:where_to_guess} Since $A_0, B_0$ are non-singular,
  then, according to \eqref{eq:chain_rule_Jacqueline}, for any
  $x \in \F_2^n$, the matrices $\Jlin{G}{x}$ and $\Jlin{F}{B_0 x}$
  have the same rank.
\end{enumerate}

\subsubsection{Sketch of the Algorithm}
The search for the pair $(A_0, B_0)$ will be done by trying to guess
pairs $(v_1, w_1), \dots, (v_s, w_s)$ of elements of
$\F_2^n \times \F_2^n$ such that for any $i$, $B_0 v_i = w_i$.  For
each such guess, we solve a concatenation of systems of the form
\eqref{eq:main_linear_system} and check whether it leads to a
solution. If not, we try with another
guess. 

Therefore, the complexity of the algorithm is
directly related to the number of guesses we will have to make,
which should be as small as possible. This motivates the following discussion.
\begin{itemize}
\item Using~\eqref{item:where_to_guess}, any guess $(v_i, w_i)$ should
  be chosen so that $\rank \Jlin{G}{v_i} = \rank \Jlin{F}{w_i}$.
  Therefore, the search is much easier if we search for
  elements $v \in \F_2^n$ (resp. $w \in \F_2^n$) such that
  $\rank \Jlin{G}{v}$ (resp. $\rank \Jlin{F}{w}$) occurs
  rarely in the rank table of $\Jlin{G}{x}$
  (resp. $\Jlin{F}{x}$). This motivates the study of this rank table
  in Section~\ref{section:diff}.
\item While the number $s$ of simultaneous guesses
  $(v_1, w_1), \dots, (v_s, w_s)$ should be the smallest possible in
  order to minimize the time complexity, it
  should also be large enough so that the linear
  system\if\IEEEversion1
  \eqref{eq:linear_system_with_s_guesses}
  \begin{figure*}[h]
    \centering
\begin{equation}
    \label{eq:linear_system_with_s_guesses}
    \left\{
      \begin{array}{ccc}
        X \cdot \Jlin{G}{v_i} - \Jlin{F}{w_i} \cdot Y & = & 0 \\
        Y \cdot v_i & = & w_i
      \end{array}
    \right.\qquad  \forall i \in \{1, \dots, s\},
  \end{equation}
  \end{figure*}
  \else:
  \begin{equation}
    \label{eq:linear_system_with_s_guesses}
    \left\{
      \begin{array}{ccc}
        X \cdot \Jlin{G}{v_i} - \Jlin{F}{w_i} \cdot Y & = & 0 \\
        Y \cdot v_i & = & w_i
      \end{array}
    \right.\qquad  \forall i \in \{1, \dots, s\},
  \end{equation}
  \fi
  has a unique solution, or a ``small enough'' affine space of
  solutions. Thus, the rank of such a system, which is nothing but a
  concatenation of $s$ systems of the form (\ref{eq:main_linear_system}),
  is investigated in
  Section~\ref{sec:rank_of_the_system}.
\item As soon as a possibly valid pair $(A_0, B_0)$ is found, it
  remains to recover $C_0, a$. This can be done with elementary linear algebra
  and is detailed in Section~\ref{ss:finishing}.
\end{itemize}

\subsubsection{Rank Table and Connection with the Differential Spectrum}\label{section:diff}

As explained earlier, part of the algorithm consists in guessing an
$s$--tuple of vectors in $\F_2^n$, which requires up to $2^{sn}$
trials with brute force. This number of trials can be
drastically improved using the {\em rank tables}.

\begin{definition}\label{def:rank table}
  The {\em rank table} $\rktable{F}$ of a function $F$ from \(\F_2^n\)
  to \(\F_2^m\), is a table with
  $(\min(m, n) + 1)$ entries indexed by $\{0, \dots, \min(m,n)\}$
  and\if\IEEEversion1,
  for any $j \in \{0, \dots, \min (m,n)\}$,
  \fi
  \[
  \if\IEEEversion0  
  \forall j \in \{0, \dots, \min (m,n)\}, \quad
  \fi
  \rktable{F}[j] := \{x \in \F_2^n ~|~ \rank (\Jlin{F}{x}) = j\},
\]
where the rank of a linear mapping is defined as the dimension of its
image set.

The {\em rank distribution} $\rkdist{F}$ of $F$ is defined as
\[
  \if\IEEEversion0
  \forall j \in \{0, \dots, \min (m,n)\},\quad
  \fi
  \rkdist{F}[j] := \# \rktable{F}[j]\if\IEEEversion1,
  \else.
  \fi
\]
\if\IEEEversion1
for any $j \in \{0, \dots, \min (m,n)\}$.
\fi
\end{definition}

\begin{lemma}\label{lem:cost_rk_table}
  The computation of the rank table can be performed in
  $\bigoh{2^n\max(m,n)^{\omega}}$ binary operations,
  where $\omega$ denotes the complexity exponent of matrix multiplication.
\end{lemma}

\begin{proof}
  First, we suppose that functions are represented by their truth
  tables.  If they are represented by their ANFs, the translation
  into a truth table costs $\bigoh{mn2^n}$ (see
  \cite[Chap.~9]{add:Joux09}), which is negligible compared to
  the final complexity.

  The computation of one entry of the rank table, corresponding to a
  vector $x \in \F_2^n$, consists in first computing $\Jlin{F}{x}$ and
  then computing the rank of this matrix.  If the truth table of
  $F$ is known, the computation of $\Jlin{F}{x}$ is done in
  $\bigoh{mn}$ binary operations according to~(\ref{eq:J}). The latter
  cost is negligible compared to that of the calculation of the rank
  which is $\bigoh{\max(m,n)^{\omega}}$ binary operations.  This
  should be repeated for any $x \in \F_2^n$, which yields the
  complexity formula.
\end{proof}

\begin{remark}
  It is worth noting that, by {\em rank table}, we refer to an object
  which is different from that used by Dinur in \cite{EC:Dinur18},
  where another ``rank table'' is used to decide affine
  equivalence. Dinur's ``rank table'' is obtained by considering
  the {\em  symbolic ranks} of $F, G$ \cite[Section~2, page 418]{EC:Dinur18},
  which are not related to the ranks of the transformations involved
  in Definition~\ref{def:rank table}.
\end{remark}

We now show that there is a one-to-one correspondence between the rank
distribution of a quadratic function and the distribution of its
DDT.

\begin{proposition}\label{prop:diff}
  Let $F : \F_2^{n} \to \F_2^{m}$ be a quadratic function and
  \(\rkdist{F}\) be its rank distribution. Let \(\ddist{F}\) denote
  the multiplicities of the elements in the differential spectrum of \(F\), i.e.,
  \begin{equation*}
    \ddist{F}[k] = \#\big\{ (a, b) \in (\ftwo^{n})^{2} ~|~
    \delta_F(a,b)=k \big\} ~.
  \end{equation*}
Then, for any \(r\), \(0 \leq r \leq \min(m,n)\),
  \[\rkdist{F}[r] = 2^{-r}\ddist{F}[2^{n-r}].\]
\end{proposition}
\begin{proof}
  Let \(a \in \F_2^n\). Proposition~\ref{prop:derivative} implies that
  \if\IEEEversion1
  \begin{align*}
    \Jlin{F}{a} \cdot x & =  \Jlin{F}{x} \cdot a\\
                        & =  \Delta_aF(x) + \Delta_aF(0).
    \end{align*}
  \else
  \[\Jlin{F}{a} \cdot x =  \Jlin{F}{x} \cdot a = 
    \Delta_aF(x) + \Delta_aF(0).\]
  \fi
  Then, the elements
  \(x \in \F_2^n\) in the right kernel of \(\Jlin{F}{a}\), i.e., in
  the set of all \(x\) such that \(\Jlin{F}{a} \cdot x = 0\), are those
  such that \(\Delta_aF(x) + \Delta_aF(0) =0\). Since \(F\) is
  quadratic, this second set is a linear space whose dimension
  equals~\(i\) where
\[\delta_F(a,b)\in \{0, 2^i\}, \; \forall b \in \F_2^n.\]
Let \(r\) denote the rank of \(\Jlin{F}{a}\).
We then
deduce that
\if\IEEEversion1
\begin{align*}
  \# \{x \in &\F_2^n  : \Jlin{F}{a} \cdot x = 0\}\\
  & =  2^{n-r} \\
  & =   \# \{x \in \F_2^n : \Delta_aF(x) + \Delta_aF(0)=0\} \\
  & =  2^i,
\end{align*}
\else
\begin{eqnarray*}
 \# \{x \in \F_2^n  : \Jlin{F}{a} \cdot x = 0\} & = & 2^{n-r} \\
  & = &  \# \{x \in \F_2^n : \Delta_aF(x) + \Delta_aF(0)=0\} \\
 & = & 2^i,
\end{eqnarray*}
\fi
where \(\{\delta_F(a,b) : b \in \F_2^m\} = \{0, 2^i\}\).
Then, the
entries in the row of the DDT indexed by \(a\) are \(0\) and
\(2^{n-r}\), and the value \(2^{n-r}\) appears \(2^{r}\) times. It
follows that
\[\rkdist{F}[r] = 2^{-r}\ddist{F}[2^{n-r}].\]
\end{proof}
\begin{remark}
It is worth noting that
the previous proposition implies that, for any quadratic function
$F$, we have $\rkdist{F}[n] =
0$, since the values in the differential spectrum of a Boolean
function are always even.
\end{remark}

As another consequence, we get that the differential spectrum of a
quadratic function contains two values only, \(0\) and \(\delta(F)\),
if and only if the matrices \(\Jlin{F}{x}\) for all nonzero
\(x \in \F_2^n\) have the same rank. This includes for instance the
case of bent functions from \(\F_2^n\) to \(\F_2^m\), \(n \geq 2m\),
and the case of APN functions. When \(m=n\), this can be seen as a
restatement of the QAM property from~\cite[Th.~1]{add:YuWanLi14},
using the correspondence exhibited in Prop.~\ref{prop:QAM}.
\begin{corollary}\label{cor:rank_table_APN}
  Let $F : \ftwo^{n} \to \ftwo^{m}$ be a quadratic function. Then,
  \(F\) is APN if and only if \(\Jlin{F}{x}\) has rank \((n-1)\) for
  all nonzero \(x\).
\end{corollary}

\subsubsection{On the Rank of the Linear
  System~\eqref{eq:main_linear_system}}
\label{sec:rank_of_the_system}
\begin{proposition}\label{prop:upper_bound_for_ranks_of_systems}
  Let $F, G$ be two quadratic functions and $v, w \in \F_2^n$ be such
  that \[\rank \Jlin{G}{v} = \rank \Jlin{F}{w}.\] Denote by $r$ the
  above rank. Then the first part of 
  System~\eqref{eq:main_linear_system}, {\em i.e.} the linear system
with unknowns
  $(X, Y) \in \F_2^{m \times m} \times \F_2^{n \times n}$:
  \begin{equation}\label{eq:linear_equations_on_A_B}
    X \cdot \Jlin{G}{v} - \Jlin{F}{w} \cdot Y = 0
  \end{equation}
  has $m^2 + n^2$ unknowns in \(\F_2\), $mn$ equations and rank less than or equal
  to $r(m+n-r)$.
  Moreover, the affine system \eqref{eq:main_linear_system}
  has $m^2 + n^2$ unknowns, $(m+1)n$ equations and rank less than
  or equal to $r(m+n-r) + (n-r)$.
\end{proposition}

\begin{proof}
  The number of unknowns corresponds to the number of entries
  of $X, Y$. The number of linear equations equals the number of
  entries of the resulting matrix $X\cdot \Jlin{G}{v} - \Jlin{F}{w}
  \cdot Y$ which is $mn$. Let us investigate the rank.
  Since $\Jlin{G}{v}$ has rank $r$, there exists a non-singular
  matrix $U \in \F_2^{n\times n}$ such that the $n-r$ rightmost
  columns of $\Jlin{G}{v}\cdot U$ are zero. Hence it has the following
  shape: 
  \begin{center}
    \begin{tikzpicture}
      \node at (0,0) {\( \Jlin{G}{v} \cdot U =
        \begin{pmatrix}
          {\rm J}_1G & 0 \\
          {\rm J}_2G & 0
        \end{pmatrix}
        \)};
      \draw[<->] (.55, -.55) -- (1.3, -.55);
      \draw[<->] (1.45, -.55) -- (2.15, -.55);
      \draw[<->] (2.35, -.4) -- (2.35, -.05);
      \draw[<->] (2.35, 0) -- (2.35, .4);
      \node at (2.6, .2) {${\scriptstyle r}$};      
      \node at (2.82, -.2) {${\scriptstyle m-r}$};      
      \node at (.93, -.73) {${\scriptstyle r}$};      
      \node at (1.8, -.73) {${\scriptstyle n-r}$};      
    \end{tikzpicture}
  \end{center}
  for some matrices ${\rm J}_1G, {\rm J}_2G$.
Similarly, there exists a
  non-singular matrix $V \in \F_2^{m \times m}$ such that
  \( V \cdot \Jlin{F}{w} \) has the following shape:
  \begin{center}
    \begin{tikzpicture}
      \node at (0,0) {\( V \cdot \Jlin{F}{w} =
        \begin{pmatrix}
          {\rm J}_1F & {\rm J}_2 F \\
          0 & 0
        \end{pmatrix}
        \)};
      \draw[<->] (.55, -.55) -- (1.3, -.55);
      \draw[<->] (1.45, -.55) -- (2.25, -.55);
      \draw[<->] (2.5, -.4) -- (2.5, -.05);
      \draw[<->] (2.5, 0) -- (2.5, .4);
      \node at (2.7, .2) {${\scriptstyle r}$};      
      \node at (2.95, -.2) {${\scriptstyle m-r}$};      
      \node at (.93, -.73) {${\scriptstyle r}$};      
      \node at (1.85, -.73) {${\scriptstyle n-r}$};      
    \end{tikzpicture}
  \end{center}
  for some matrices ${\rm J}_1F, {\rm J}_2F$.
  Thus, setting $X' := VX$ and $Y' := YU$, we get a new and
  equivalent linear system
  \begin{equation}\label{eq:equiv_linear system}
    X' \cdot \Jlin{G}{v} \cdot U - V \cdot \Jlin{F}{v} \cdot Y' = 0.
  \end{equation}
  If we denote the block decompositions of $X', Y'$ as
  \if\IEEEversion1
  \begin{center}
    \begin{tikzpicture}
      \node at (0,0) {\( X' =
        \begin{pmatrix}
          X'_1 & X'_2 \\ X'_3  & X'_4
        \end{pmatrix}
        \)};
      \node at (1.6, .2) {${\scriptstyle r}$};      
      \node at (1.85, -.2) {${\scriptstyle m-r}$};      
      \node at (0.05, -.73) {${\scriptstyle r}$};      
      \node at (.85, -.73) {${\scriptstyle m-r}$};      
      \draw[<->] (-.3, -.55) -- (.4, -.55);
      \draw[<->] (.5, -.55) -- (1.2, -.55);
      \draw[<->] (1.4, -.4) -- (1.4, -.05);
      \draw[<->] (1.4, 0) -- (1.4, .4);
    \end{tikzpicture}
    \end{center}
    and
    \begin{center}
          \begin{tikzpicture}
      \node at (5,0) {\(Y' =
        \begin{pmatrix}
          Y'_1 & Y'_2 \\ Y'_3 & Y'_4
        \end{pmatrix}
        \)};
      \draw[<->] (4.7, -.55) -- (5.4, -.55);
      \draw[<->] (5.5, -.55) -- (6.2, -.55);
      \draw[<->] (6.4, -.4) -- (6.4, -.05);
      \draw[<->] (6.4, 0) -- (6.4, .4);
      \node at (6.6, .2) {${\scriptstyle r}$};      
      \node at (6.85, -.2) {${\scriptstyle n-r}$};      
      \node at (5.05, -.73) {${\scriptstyle r}$};      
      \node at (5.85, -.73) {${\scriptstyle n-r}$}; 
     \end{tikzpicture}
   \end{center}
  \else
  \begin{center}
    \begin{tikzpicture}
      \node at (0,0) {\( X' =
        \begin{pmatrix}
          X'_1 & X'_2 \\ X'_3  & X'_4
        \end{pmatrix}
        \)};
      \node at (2.8, 0) {and};
      \node at (5,0) {\(Y' =
        \begin{pmatrix}
          Y'_1 & Y'_2 \\ Y'_3 & Y'_4
        \end{pmatrix}
        \)};
      \draw[<->] (-.3, -.55) -- (.4, -.55);
      \draw[<->] (.5, -.55) -- (1.2, -.55);
      \draw[<->] (1.4, -.4) -- (1.4, -.05);
      \draw[<->] (1.4, 0) -- (1.4, .4);
      \node at (1.6, .2) {${\scriptstyle r}$};      
      \node at (1.85, -.2) {${\scriptstyle m-r}$};      
      \node at (0.05, -.73) {${\scriptstyle r}$};      
      \node at (.85, -.73) {${\scriptstyle m-r}$};      
      \draw[<->] (4.7, -.55) -- (5.4, -.55);
      \draw[<->] (5.5, -.55) -- (6.2, -.55);
      \draw[<->] (6.4, -.4) -- (6.4, -.05);
      \draw[<->] (6.4, 0) -- (6.4, .4);
      \node at (6.6, .2) {${\scriptstyle r}$};      
      \node at (6.85, -.2) {${\scriptstyle n-r}$};      
      \node at (5.05, -.73) {${\scriptstyle r}$};      
      \node at (5.85, -.73) {${\scriptstyle n-r}$}; 
     \end{tikzpicture}
   \end{center}
   \fi
  then the block decomposition of 
  System~\eqref{eq:equiv_linear system} gives\if\IEEEversion1
  \eqref{eq:block_decomp}.
  \begin{figure*}[!h]
    \begin{equation}\label{eq:block_decomp}
    \begin{tikzpicture}
      \node at (0,0) {\(
        \begin{pmatrix}
          X_1' {\rm J}_1G +  X_2' {\rm J}_2G + {\rm J}_1F Y_1' + {\rm J}_2F Y_3'
          & {\rm J}_1 F Y_2' + {\rm J}_2 F Y_4' \\
          X_3' {\rm J}_1 G + X_4' {\rm J}_2 G & 0
        \end{pmatrix} =
        \begin{pmatrix}
          0 & 0 \\ 0 & 0
        \end{pmatrix}
        \)};
      \draw[<->] (-4.9, -.55) -- (0.67, -.55);
      \draw[<->] (0.72, -.55) -- (3.2, -.55);
      \draw[<->] (-5.55, -.4) -- (-5.55, -.05);
      \draw[<->] (-5.55, 0) -- (-5.55, .4);
      \node at (-5.75, .2) {${\scriptstyle r}$};      
      \node at (-6, -.2) {${\scriptstyle m-r}$};      
      \node at (-2, -.73) {${\scriptstyle r}$};      
      \node at (2, -.73) {${\scriptstyle n-r}$};      
      \draw[<->] (3.8, -.55) -- (4.5, -.55);
      \draw[<->] (4.6, -.55) -- (5.4, -.55);
      \draw[<->] (5.6, -.4) -- (5.6, -.05);
      \draw[<->] (5.6, 0) -- (5.6, .4);
      \node at (5.8, .2) {${\scriptstyle r}$};      
      \node at (6, -.2) {${\scriptstyle m-r}$};      
      \node at (4.15, -.73) {${\scriptstyle r}$};      
      \node at (5., -.73) {${\scriptstyle n-r}$};   
   \node at (6.4, 0) {.}; 
 \end{tikzpicture}
      \end{equation}
      \end{figure*}
   \else:
  \begin{center}
    \begin{tikzpicture}
      \node at (0,0) {\(
        \begin{pmatrix}
          X_1' {\rm J}_1G +  X_2' {\rm J}_2G + {\rm J}_1F Y_1' + {\rm J}_2F Y_3'
          & {\rm J}_1 F Y_2' + {\rm J}_2 F Y_4' \\
          X_3' {\rm J}_1 G + X_4' {\rm J}_2 G & 0
        \end{pmatrix} =
        \begin{pmatrix}
          0 & 0 \\ 0 & 0
        \end{pmatrix}
        \)};
      \draw[<->] (-4.9, -.55) -- (0.67, -.55);
      \draw[<->] (0.72, -.55) -- (3.2, -.55);
      \draw[<->] (-5.55, -.4) -- (-5.55, -.05);
      \draw[<->] (-5.55, 0) -- (-5.55, .4);
      \node at (-5.75, .2) {${\scriptstyle r}$};      
      \node at (-6, -.2) {${\scriptstyle m-r}$};      
      \node at (-2, -.73) {${\scriptstyle r}$};      
      \node at (2, -.73) {${\scriptstyle n-r}$};      
      \draw[<->] (3.8, -.55) -- (4.5, -.55);
      \draw[<->] (4.6, -.55) -- (5.4, -.55);
      \draw[<->] (5.6, -.4) -- (5.6, -.05);
      \draw[<->] (5.6, 0) -- (5.6, .4);
      \node at (5.8, .2) {${\scriptstyle r}$};      
      \node at (6, -.2) {${\scriptstyle m-r}$};      
      \node at (4.15, -.73) {${\scriptstyle r}$};      
      \node at (5., -.73) {${\scriptstyle n-r}$};   
   \node at (6.4, 0) {.}; 
        \end{tikzpicture}
      \end{center}
      \fi
      Each of the $mn$ entries of the left-hand matrix\if\IEEEversion1
      of \eqref{eq:block_decomp}
      \fi
      \ yields a linear
  equation relating the entries of $X', Y'$.  Hence, one can ignore
  the entries of the bottom right--hand corner yielding equations of
  the form $0 = 0$, which leaves $r(m+n-r)$ equations. This yields the upper bound
  on the rank of System~\eqref{eq:equiv_linear system} and hence that
  of System~\eqref{eq:linear_equations_on_A_B}.

  Concerning the whole system~\eqref{eq:main_linear_system}, including
  the affine equations $Y \cdot v = w$ amounts to considering $n$
  additional affine equations. However, these equations are never
  independent. Indeed, if $(X, Y)$ is a solution of
  (\ref{eq:linear_equations_on_A_B}), then
  $X \Jlin{G}{v} = \Jlin{F}{w} Y$ and, according to
  Corollary~\ref{cor:a_in_kerJ(a)}, we deduce that $Y\cdot v$ should
  lie in the right kernel of $\Jlin{F}{w}$ which has dimension $n-r$.
  Thus, the $n$ additional affine equations on the entries of $Y$
  given by $Y\cdot v = w$ impose at most $n-r$ new independent conditions
  on the entries of $Y$.
\end{proof}

Using SageMath \cite{add:sage}, we ran about 500 tests for various
values of $n,m \in \{6,8\}$. For each test, we considered the cases
where $F, G$ are EA-equivalent, and when they are not. To do so $F$ was drawn
uniformly at random among quadratic functions, then $G$ was either
drawn uniformly at random or a triple $(A, B, C)$ was drawn at random and $G$
was defined as $A \circ F \circ B + C$.
Our experimental observations
showed that for a fixed pair $(v, w)$ the upper bounds on the ranks of
(\ref{eq:main_linear_system}) and (\ref{eq:linear_equations_on_A_B})
were always reached. Thus, we claim that these bounds are sharp.
On the
other hand, when considering $s$--tuples of pairs $(v_1, w_1), \dots,
(v_s, w_s)$, the concatenation of $s$ systems of the form
(\ref{eq:main_linear_system}), {\em i.e.} a system of the form
(\ref{eq:linear_system_with_s_guesses}), does not in general lead to a
system of rank
\begin{equation}\label{eq:upper_bound_rank}
  \sum_{i=1}^s \left(r_i(m+n-r_i)+(n-r_i)\right),
\end{equation}
but to a system of slightly smaller rank and we did not succeed in getting a
sharper estimate as soon as $s > 1$, which is actually the use-case of our
algorithm.
These experimental results are summarized in
  Table~\ref{tab:rank_experiments}.  Our experimental observations show
  that for many use-cases, the choice $s = 3$ is appropriate since it leads
  to a system with few solutions.
  \if\IEEEversion1
  \begin{table*}[h!tb]
  \else
  \begin{table}[h!tb]
  \fi
  \centering
  \renewcommand\arraystretch{1.3}
  \setlength\tabcolsep{4pt}
  \tiny
  \begin{tabular}{clccccc}
    \toprule
    {\normalsize $m$} & {\normalsize $n$} & {\normalsize $m^2+n^2$} & {\normalsize $s$} & {\normalsize
    Jacobian ranks} &
    {\normalsize Expected ranks for (\ref{eq:linear_equations_on_A_B})
    and (\ref{eq:main_linear_system})}  &
    {\normalsize Observed ranks (intervals)} \\
    \midrule
    6 & 6 & 72 & 1 & (3) & (27, 30) & (27, 30) \\
    \rowcolor{gray!10}
    6 & 6 & 72 & 1 & (4) & (32, 34) & (32, 34) \\
    6 & 8 & 100 & 1 & (4) &  (40, 44) & (40, 44) \\
    \rowcolor{gray!10}
    8 & 6 & 100 & 1 & (4) & (40, 42) & (40, 42) \\
    \midrule
    6 & 6 & 72 & 2 & (3,3) & (54, 60) & (48$\dots$53, 50$\dots$54) \\
    \rowcolor{gray!10}
    6 & 6 & 72 & 2 & (3,4) & (59, 64) & (52$\dots$56, 56$\dots$57) \\
    6 & 6 & 72 & 2 & (4,4) & (64, 68) & (56$\dots$60, 60$\dots$61) \\
    \rowcolor{gray!10}
    6 & 8 & 100 & 2 & (4,4)  & (80, 88) & (72$\dots$80, 78$\dots$82) \\
    8 & 6 & 100 & 2 & (4,4) & (80, 84) & (68$\dots$72, 72$\dots$73) \\
    \midrule
    \rowcolor{gray!10}
    6 & 6 & 72 & 3 & (3,4,4) & (72, 72) & (65$\dots$72, 69$\dots$72) \\    
    6 & 6 & 72 & 3 & (4,4,4) & (72, 72) & (60$\dots$72, 66$\dots$72) \\    
    \rowcolor{gray!10}
    6 & 8 & 100 & 3 & (4,4,4) & (100, 100) & (90$\dots$100, 93$\dots$100) \\    
    8 & 6 & 100 & 3 & (4,4,5) & (100, 100) & (95$\dots$96, 96) \\    
    \bottomrule
  \end{tabular}
  \caption{This table summarizes some experimental results on about
    500 tests.  The column {\em Jacobian ranks} gives the ranks of
    $\Jlin{G}{v_i}$ (or equivalently $\Jlin{F}{w_i}$) for the
    $(v_i, w_i)$'s we tested. The column {\em Expected ranks}
    corresponds to the previously obtained upper bounds on the ranks
    of the two systems given in
    Proposition~\ref{prop:upper_bound_for_ranks_of_systems}.  The
    column {\em Observed ranks}, provides the experimentally observed
    ranks for these systems. They correspond to the expected ones for $s=1$,
    while for higher $s$ they can be smaller and vary in the given intervals.
  }
  \label{tab:rank_experiments}
  \if\IEEEversion1
\end{table*}
\else
\end{table}
\fi

For APN functions, the smallest rank occurring in the rank
distribution is $n-1$. In this situation the rank of a system of the
form (\ref{eq:linear_equations_on_A_B}) is larger and in such a case,
the choice $s = 2$ might be sufficient, as suggested by the following
example.

\begin{example}\label{ex:rank_of_system_for_APN}
  Consider the functions \(F\) from \(\F_{2^n}\) to \(\F_{2^n}\)
  defined by \(F(x) = x^{2^i+1}\), with $i$ coprime to $n$. These
  functions, known as the Gold functions, form one the earliest known
  families of APN functions. When regarded as vectorial Boolean
  functions $\F_2^n \rightarrow \F_2^n$, they are quadratic.  For
  $n = 7$, choosing $s = 2$ yields solution spaces of dimension $7$
  while for $s = 3$ these spaces have dimension at most~$1$.
  Therefore, for such functions, it seems unclear which choice for $s$
  is the most suitable. On the one hand, choosing $s = 2$ will make it
  necessary to perform a brute-force search over a set of size
  $\bigoh{2^{2n}}$ but each step will require a second brute-force in
  the solution space of
  System~\eqref{eq:linear_system_with_s_guesses}, which is reasonable
  for $n = 5, 7$ but non negligible. On the other hand, the
  brute-force search is performed on a set of guesses of
  $\bigoh{2^{3n}}$ elements, but each step of the search is much less
  expensive.  These cases have been practically compared and, in terms
  of running times, the choice $s = 2$ seems to be more favorable (see
  Table~\ref{tab:algo_APN_tests} in Section~\ref{subsec:implem}).
\end{example}

\subsubsection{Deducing the Full Extended-Affine
  Equivalence}\label{ss:finishing}
Once a pair $(A_0, B_0)$ is computed, the remainder of the $4$--tuple \((A_0,B_0,C_0,a)\)
defining EA-equivalence,
can be computed as follows. Let $G_1$ be the function defined by
$\ G_1(x) := A_0 F (B_0 x)$ and recall that
$G(x) = A_0 F(B_0 x) + C_0 x + a$. Since $G_1$ can be computed from
the triple $(F, A_0, B_0)$, then, one gets $a$ using
\[
  a = G(0) + G_1(0).
\]
Finally, $C_0$ can be computed as the matrix
satisfying
\[
  \forall x\in \F_2^n, \quad C_0 x = G(x) + G_1(x) + a.
\]

\subsubsection{The Algorithm}\label{subsec:algo}
The pseudo-code of the full procedure for recovering extended-affine
equivalence is summarized in Algorithm~\ref{algo:get_EA_equivalence}.
Below we give a description of the successive steps of the algorithm with several comments.
\begin{enumerate}
\item Compute the rank tables and the rank distributions of $\Jlin{F}{x}$
  and $\Jlin{G}{x}$. If these distributions differ, then the functions
  are not EA-equivalent.
\item Otherwise, estimate a reasonable number of guesses $s$ yielding
  Systems~\eqref{eq:linear_system_with_s_guesses} with few
  solutions.  For many parameters, the choice $s = 3$ turns out to be
  appropriate.
\item Choose $s$ {\em reference vectors} $w_1, \dots, w_s \in \F_2^n$
  for which the values $\rank \Jlin{F}{w_i}$ lie among the least
  frequently occurring values in the rank distribution.
\item Perform a brute-force search over all $s$--tuples $(v_1, \dots, v_s)$
  such that $\rank \Jlin{G}{v_i} = \rank \Jlin{F}{w_i}$ for any
  $i \in \{1, \dots, s\}$.  For each such guess, solve the
  system\if\IEEEversion1
  \eqref{eq:new_system}.
  \begin{figure*}[!h]
    \centering
    \begin{equation}
      \label{eq:new_system}
      \left\{
          \begin{array}{ccc}
            X \cdot \Jlin{G}{v_i} -  \Jlin{F}{w_i}\cdot Y & = & 0\\
            Y \cdot v_i & = & w_i, \quad \forall i \in \{1, \dots, s\}.
          \end{array}
        \right.
    \end{equation}
  \end{figure*}
  \else:
      \[
        \left\{
          \begin{array}{ccc}
            X \cdot \Jlin{G}{v_i} -  \Jlin{F}{w_i}\cdot Y & = & 0\\
            Y \cdot v_i & = & w_i, \quad \forall i \in \{1, \dots, s\}.
          \end{array}
        \right.
       \]
       \fi
     \item If the above system has ``too many solutions'', make
       another guess.  A threshold $T$ for the dimension of the space
       of solutions should have been chosen in advance. In our
       experiments, we set it to $10$ in order to have at most $1024$
       solutions for the system. The relevance of the choice for $T$
       is discussed further in Remark~\ref{rem:variant}.
       If the threshold is exceeded too
       frequently, the value of~$s$ may be underestimated and restarting the
       algorithm with a larger $s$ would be appropriate.
    \item For each guess for which the space of solutions of the
      system is small enough, we perform brute-force search in this
      solution space for a pair of matrices $(X, Y)$ which are both
      non-singular. Such a pair provides a candidate for
      $(A_0^{-1}, B_0)$. For such a candidate, we use the calculations
      of Section~\ref{ss:finishing} to deduce a $4$--tuple
      $(A_0, B_0, C_0, a)$. If such a $4$--tuple is found, the problem
      is solved, if
      not, we keep on searching.
\end{enumerate}

\if\IEEEversion1
\begin{algorithm*}[!h]
\else
\begin{algorithm}
\fi
  \begin{itemize}
  \item {\bf Input: } A pair of vectorial Boolean functions
    $F,G : \F_2^n \longrightarrow \F_2^m$ and a threshold $T$.\item {\bf Output: } A 4-tuple $(A_0, B_0, C_0, a)$ such that
    \[\forall x \in \F_2^n,\ G(x) = A_0 \cdot F(B_0 x) + C_0x +a\] if it exists.
    Otherwise, returns ``NOT EQUIVALENT'' or ``NO EQUIVALENCE FOUND''.
  \end{itemize}
  \vspace{-.4cm}
  \hrulefill
  \begin{algorithmic}[1]
    \State{Compute $\rktable{F}$ and $\rktable{G}$ and the
      corresponding rank distributions.}
    \If{$\rkdist{F} \neq \rkdist{G}$} \State{\Return{``NOT
        EQUIVALENT''}} \EndIf \State{Let $r_0$ be one of the least frequently occurring
      nonzero ranks.}  \State{Determine the number $s$ of vectors
      to guess (in general $s=3$ is enough).}  \State{Choose $s$
      linearly independent reference vectors
      $w_1,\dots, w_s \in \rktable{F}[r_0]^s$.}  \For{any $s$--tuple of
      linearly independent vectors
      $(v_1, \dots, v_s) \in \rktable{G}[r_0]^s$}\label{forloop}
    \State{Compute the solution space $\mathcal S$ of the system with
      variables $(X, Y)\in \F_2^{m\times m} \times \F_2^{n \times n}$
      \[
        \left\{
          \begin{array}{ccc}
            X \cdot \Jlin{G}{v_i} -  \Jlin{F}{w_i}\cdot Y & = & 0\\
            Y \cdot v_i & = & w_i.
          \end{array}
        \right.,
        \quad \forall i \in \{1, \dots, s\}.
      \]
    }
    \If{$\dim \mathcal S \leq T$}
    \For{$(X, Y) \in \mathcal S$}
    \If{both $X$ and $Y$ are non-singular}
    \State{$A_0 := X^{-1}$ and $B_0 := Y$}
    \If{Some $4$--tuple $(A_0, B_0, C_0, a)$ may be deduced using
      Section~\ref{ss:finishing}}
    \State{\Return $(A_0, B_0, C_0, a)$}
    \EndIf
    \EndIf
    \EndFor
    \EndIf
    \EndFor
    \State{\Return{``NO EQUIVALENCE FOUND''}}
  \end{algorithmic}
  \caption{Algorithm for EA-equivalence recovery.}
  \label{algo:get_EA_equivalence}
\if\IEEEversion1
\end{algorithm*}
\else
\end{algorithm}
\fi

\begin{remark}
  The reference vectors $(w_1, \ldots, w_s)$ should be chosen to be linearly
  independent.  Indeed, if some $w_i$ can be expressed as a linear combination of
  the others, then it does not contribute anything to the linear
  system (\ref{eq:linear_system_with_s_guesses}), since the
  corresponding equations can be derived from the other ones.
\end{remark}

\begin{remark}\label{rk:hybrid_cases}
  It may happen that $\rkdist{F}[r_0]$ is smaller than $s$, or more
  generally that $\rktable{F}[r_0]$ does not contain $s$ linearly
  independent elements. This situation is actually advantageous:
  assume for instance that $s = 3$ and there are only $2$ vectors
  $w_1, w_2$ in $\rktable{F}[r_0]$ and 2 vectors $v_1, v_2$ in
  $\rktable{G}[r_0]$. In this situation, we choose $r_1$ which is the
  index of the smallest entry larger than $\rkdist{G}[r_0]$ in the rank
  distribution, choose $v_3 \in \rktable{G}[r_1]$ and perform a
  brute-force search for a $w_3 \in \rktable{F}[r_1]$ such that either
  $((v_1, w_1), (v_2, w_2), (v_3, w_3))$ or
  $((v_1, w_2), (v_2, w_1), (v_3, w_3))$ is an appropriate guess.  The
  number of guesses we should investigate is at most
  $2 \rkdist{F}[r_1]$ instead of $\rkdist{F}[r_0]^3$.
\end{remark}

\begin{remark}\label{rem:variant}
  In the way it is described, the algorithm might fail to return the
  solution while the functions are EA-equivalent.  Indeed, if for a
  given guess $(v_1, \ldots, v_s)$,
  System~\eqref{eq:linear_system_with_s_guesses} has a solution space
  whose dimension exceeds the threshold $T$, this guess is not
  investigated further. For this reason, we might fail to find the
  \(4\)--tuple giving the EA-equivalence.  This also raises the
  question of the most relevant choice for the pair $(s,T)$.

  Let us first discuss how to appropriately choose $(s,T)$.
  Actually, one observes that the dimension of the solution space of
  System~(\ref{eq:linear_system_with_s_guesses}) almost always lies in
  a given range that can be estimated by running a few tests.  By this
  manner, one can fix a bound $d$ such that the dimension of this
  solution space is almost always less than or equal to $d$.  Then,
  the complexity of the algorithm is $\bigoh{(m^2+n^2)^\omega 2^dR^{s}}$.
  If $d \leq \log_2(R)$, then the choice of $s$ is appropriate, otherwise it
  is probably better to increment $s$. As soon as $s$ is fixed and $d$
  is estimated, a reasonable value for $T$ would be $d$ or a bit more
  in order to take into account rare cases when the dimension of the
  solution space of System~(\ref{eq:linear_system_with_s_guesses}) is
  larger.  For all the tests we ran, we had $n < 10$. Since we always
  have $R < 2^n$, taking $T = 10$ seemed to be a reasonable
  choice because the quantity $d$ was always less than
  $\log_2 (R) < n < 10$.

  Finally, once $(s,T)$ is fixed,
  there is still a risk to miss some exceptional guesses for
  which the dimension of
  System~(\ref{eq:linear_system_with_s_guesses}) is larger than $T$.
  A possible way to address this issue would be to run some recursive
  call when the dimension of the solution space of
  System~(\ref{eq:linear_system_with_s_guesses}) exceeds $T$ {and to
    perform an exhaustive search for an $(s+1)$--th vector \(v_{s+1}\)
    in order to reduce the dimension of the solution space of
    System~(\ref{eq:linear_system_with_s_guesses})}.  Such an improvement
  of the algorithm, which has not been implemented, would permit to
  avoid situations where the equivalence is not detected.
\end{remark}

\subsubsection{Complexity}
According to Lemma~\ref{lem:cost_rk_table}, the cost of the
computation of the rank table is $\bigoh{\max(n,m)^\omega 2^n}$.
Next, we have to evaluate the cost of the part of the algorithm that
performs the search. For any guess $v_1, \dots, v_s$, we have to solve
a linear system of $m^2+n^2$ unknowns and $\bigoh{m^2+n^2}$ equations;
actually the choice of $s$ is done so that this system has a small
solution space and hence is represented by a matrix which is close to
square.

The number of guesses we should make is of order $\bigoh{R_1 \cdots R_s}$,
where for any $i \in \{1, \dots, s\}$, $R_i := \rkdist{F}[r_i]$ with
$r_i := \rank \Jlin{F}{w_i}$, the $w_i$'s being the reference vectors
defined in Section~\ref{subsec:algo}. In summary, for a uniformly random quadratic function $F$, denoting by
$R = \max_i \{R_i\}$, the running time of the algorithm is of order
\[\bigoh{\max(n,m)^\omega 2^n + R^s (m^2+n^2)^\omega},\]
and we recall that in general $s = 3$ is an appropriate choice.  In
practice, for random quadratic functions, the vectors providing the
minimal entries of the rank distribution are very rare, resulting in a
very fast running time of the algorithm (see Remark~\ref{rk:hybrid_cases}).

On the other hand, the situation where the algorithm is the least
efficient is when the functions are APN. Indeed, as proved in
Corollary~\ref{cor:rank_table_APN}, APN functions are precisely
the ones whose nonzero ranks are all equal to $n-1$. For such functions, the
complexity of the algorithm is 
\[\bigoh{n^{2\omega}2^{sn}}.\]
Here again, the choice $s = 3$ seems appropriate.  Note that, according
to Remark~\ref{ex:rank_of_system_for_APN}, it may be possible that the
choice $s = 2$ is sufficient, and corresponds to a number of guesses of
$2^{2n}$ at the cost of some overhead for each step of the brute-force
search.  We have not been able to estimate the asymptotic behaviour of
this overhead.

\subsubsection{Examples of Running Times}\label{subsec:implem}
The algorithm has been implemented using SageMath \cite{add:sage} and
tested on a personal machine equipped with an
Intel\textsuperscript{\textregistered} Core\textsuperscript{(TM)}
i5-8250U CPU @ 1.60GHz.  The source code is available on
GitHub.\footnote{\if\IEEEversion1
  \href{https://github.com/alaincouvreur/EA_equivalence_for_quadratic_functions}{https://github.com/alaincouvreur/EA\_equivalence\_for\_}
  \href{https://github.com/alaincouvreur/EA_equivalence_for_quadratic_functions}{quadratic\_functions}
  \else
  \href{https://github.com/alaincouvreur/EA_equivalence_for_quadratic_functions}{https://github.com/alaincouvreur/EA\_equivalence\_for\_quadratic\_functions}
  \fi
}\footnote{The variant of the algorithm mentioned in
  Remark~\ref{rem:variant} has not been implemented.}  We ran our
implementation for various cases, which are described in the sequel.

\paragraph{Random quadratic functions}
The first case we considered is when $F$ is a random quadratic
function and $G = A \circ F \circ B + C$ and $A, B, C$ are drawn at
random.  In this situation, the behaviour and the running time highly
depend on the rank distribution, as shown in
Table~\ref{tab:algo_tests}.  

\if\IEEEversion1
\begin{table*}[h!tb]
\else
\begin{table}[h!tb]
  \fi
  \centering
  \renewcommand\arraystretch{1.3}
  \setlength\tabcolsep{4pt}
  \tiny
  \begin{tabular}{ccccc}
    \toprule
    {\normalsize $m$} & {\normalsize $n$} & {\normalsize Rank distribution} &
    {\normalsize Number of tries} & {\normalsize
    Time (seconds)} \\
    \midrule
    6 & 6 & $[1, 0, 0, 2, 18, 43, 0]$ & 21 & 0.68\\
    \rowcolor{gray!10}
    6 & 6 & $[1, 0, 0, 1, 24, 38, 0]$ & 386 & 5.36\\
    6 & 6 & $[1, 0, 0, 0, 27, 36, 0]$ & 4605 & 61.1\\
    \midrule
    \rowcolor{gray!10}
    6 & 8 & $[1,0,0,0,9,96,150]$ & 127 & 15.5 \\
    6 & 8 & $[1, 0, 1, 12, 98, 144]$ & 24 & 13.8 \\
    \midrule
    \rowcolor{gray!10}
    8 & 6 & $[1, 0, 0, 0, 0, 63, 0]$ & 11067 & 195.1 \\
    8 & 6 & $[1, 0, 0, 0, 3, 60, 0]$ & 318 & 53.4 \\
    \midrule
    \rowcolor{gray!10}
    8 & 8 & $[1,0,0,0,0,6,93,156,0]$ & 95 & 20.3 \\
    8 & 8 & $[1,0,0,0,1,13,104,137,0]$ & 36 & 15.3 \\
    \bottomrule    
  \end{tabular}
  \caption{This table lists some experiments conducted on pairs of
    random quadratic functions that are EA-equivalent. Each row
    corresponds to a single experiment. All of them consisted in
    guessing $s = 3$ vectors $(v_1, v_2, v_3)$ and with threshold
    $T = 10$.  The fourth column gives the number of iterations, {\em
      i.e.} the number of times we made a guess before finding a triple
    $(v_1, v_2, v_3)$ leading to a solution.}
  \label{tab:algo_tests}
  \if\IEEEversion1
\end{table*}
\else
\end{table}
\fi

The case when $F, G$ are both drawn at random and hence
are not EA-equivalent is in general much more efficiently solved since
most of the time, their rank tables are different which permits to conclude
that the functions are not EA-equivalent.

\paragraph{APN functions}
We also considered the case already mentioned in
Example~\ref{ex:rank_of_system_for_APN} by considering $F$
to be $F(x) = x^3$ when regarded as a function from $\F_{2^n}$
into $\F_{2^n}$.

We first tested the case where $G$ is EA-equivalent to $F$ for
$n \in \{7,8,9\}$.  Since we noticed in
Example~\ref{ex:rank_of_system_for_APN} that the best choice between
$s=2$ and $s=3$ was unclear, we tested both for $n=7$ and observed
that in practice choosing $s=2$ leads to better running times. For
each set of parameters we ran $50$ tests and the average times and
number of tries are presented in Table~\ref{tab:algo_APN_tests}. Note
that, each time, the number of tries we need to perform before
guessing the good vectors is much smaller than $R^{s} = 2^{2n}$. This
is explained by the fact that $F$ has a lot of ``automorphisms'', {\em
  i.e.}  there are a lot of nontrivial triples $(A, B, C)$ such that
$F = A \circ F \circ B + C$. Consequently, there are a lot of triples
$(A, B, C)$ such that $G = A \circ F \circ B + C$ making the search much
easier.

\if\IEEEversion1
\begin{table*}[h!tb]
\else
\begin{table}[h!tb]
\fi  
  \centering
  \renewcommand\arraystretch{1.3}
  \setlength\tabcolsep{4pt}
  \tiny
  \begin{tabular}{cccc}
    \toprule
    {\normalsize $m=n$} & {\normalsize $s$} & 
    {\normalsize Average number of tries} & {\normalsize
    Average time (seconds)} \\
    \midrule
    7 & 2 & 16.8 & 11.47\\
    \rowcolor{gray!10}
    7 & 3 & 1954.72 & 34.84\\
    \midrule
    8 & 2 & 39.2 & 89.49\\
    \midrule
    \rowcolor{gray!10}
    9 & 2 & 47.83 & 211 \\
    \bottomrule    
  \end{tabular}
  \caption{This table lists some experiments conducted for the APN
    function $F(x)=x^3$ (regarded as a function
    $\F_{2^n} \rightarrow \F_{2^n}$).  Each row corresponds to 50
    experiments. The threshold was set to $T = 10$. The third column
    gives the average number of times we made a guess and the last one
    is the average running time.}
  \label{tab:algo_APN_tests}
\if\IEEEversion1
\end{table*}
\else
\end{table}
\fi

Finally, we also tested the case where $F(x) = x^3$ and $G(x) = x^5$
and hence the two functions are not EA-equivalent. In this situation,
the algorithm is much less efficient since it is necessary to explore
all the possible guesses. Therefore, the running time should
not be as variable as in the case of EA-equivalent functions (where
the algorithm stops as soon as the guess is good and hence does not
have to explore the full list of possible guesses). For this reason,
we ran only one test for $n = 7$. Our program proved the functions
were non equivalent in less than 3 hours.

 \section{EA-partitioning}
\label{sec:invariants}

In this section, unlike above, we will focus on
EA-partitioning rather than EA-recovery. EA-partitioning can be solved by combining an efficient EA-invariant and an algorithm for EA-recovery.
In the particular case of quadratic APN functions, we
introduce in Section~\ref{sec:invariants-list} a new family of invariants based
on the properties of the ortho-derivative, and we compare its performance
with the previously known invariants.

We then show how these invariants can be used for partitioning the
CCZ-class of a given function into EA-classes, and we detail two
applications. First, Section~\ref{sec:applications-kim} presents
partitions of the CCZ-classes of all the 6-bit APN quadratic functions
into EA-classes. Then, in order to compare the different invariants, we look
  at the various partitions that they define over the set of all known 8-bit
  APN functions. Our results are in
  Section~\ref{sec:applications-8bit}. As we will see, in the case of
  quadratic APN functions, the invariants based on the
  ortho-derivative are by far the most discriminating.

\subsection{Solving the EA-partitioning Problem}
\label{sec:invariants-list}
A standard approach in solving the EA-partitioning problem consists of the following two consecutive steps:
\begin{enumerate}
  \item Compute some quantities for each function that are invariant under
EA-equivalence. The set of all these quantities is then used as a
bucket label: if two distinct functions fall in different buckets,
then they cannot be EA-equivalent.
\item In the case where several functions
are in the same bucket, solve the EA-recovery problem
for each pair of functions in order to sort them in different
EA-classes.
\end{enumerate}
It is worth noting that this algorithm is correct, i.e., returns the solution, if the EA-recovery algorithm used in the second step does not give any false negatives. This is the case for instance of the variant of our EA-recovery algorithm described in Remark~31.

This general approach could be applied to other forms of equivalence,
but we focus here on the cases of CCZ- and extended affine
equivalence. The invariants discussed in this paper are summarized in
Table~\ref{tab:list-invariants}, along with their time
complexities.\footnote{The number of subspaces, the thickness
  spectrum, and the subspaces among non-bent components all rely on a
  vector space search which can be done using the algorithm
  from~\cite{AC:BonPerTia19}. However, as this algorithm is
  essentially a highly optimized tree search, its time complexity is
  hard to predict.} Their definitions, as well as details on their efficiency
and distinguishing power, can be found in the very recent survey paper
by Kaleyski~\cite{add:Kaleyski21}. Note that other \emph{ad hoc} invariants can
be defined such as the one based on the presence of permutations in
the EA-class that we apply to the CCZ-class of the Kim mapping in
Section~\ref{sec:applications-kim}. Other invariants can be defined, for instance invariants
  for code equivalence, like the automorphism and multiplier groups of
  the code associated to a vectorial Boolean
  function~\cite{add:BDKM09}. Those listed in
Table~\ref{tab:list-invariants} are the ones which apply to the largest or most
  interesting classes of functions.
\if\IEEEversion1
  \begin{table*}[h!tb]
  \else
  \begin{table}[h!tb]
    \fi
  \centering
  \renewcommand\arraystretch{1.3}
  \setlength\tabcolsep{4pt}
  \begin{tabular}{clccc}
    \toprule
    Equivalence & Invariant & Condition & Complexity & Ref. \\

    \midrule

    \multirow{6}{*}{CCZ}
    & Extended Walsh spectrum & -- & $n2^{n+m}$ & -- \\
    & Differential spectrum & -- & $2^{n+m}$ & -- \\
    & $\#$ of subspaces with $\dim = n$ in $\walshZeroes{F}$ & -- & ? & \cite{add:CanPer18} \\
    & $\Gamma$-rank & $m=n$ & $2^{\omega n}$ & \cite{add:BDKM09} \\
    & $\Delta$-rank & $m=n$ & $2^{\omega n}$ & \cite{add:BDKM09} \\
    & Distance invariant $\Pi$ & &$2^{3n}$  & \cite{add:BCHK20}\\

    \midrule
    \multirow{6}{*}{EA}

    & Algebraic degree & -- & $m 2^{n}$ & -- \\
    & Thickness spectrum & -- & ? & \cite{add:CanPer18} \\
    & $\Sigma^{k}$-multiplicities & $k, n$ even & $n2^{n+m}$ & \cite{add:Kaleyski20} \\
    & $\#$ of subspaces in non-bent components & $\deg(F)=2$ & ? & \cite{add:BCCCV21,add:GolPav21} \\
    & \multirow{2}{*}{Affine equivalence class of $\pi_{F}$} & $\deg(F)=2$, & \multirow{2}{*}{$2^{n+m}$} & \multirow{2}{*}{Sec.~\ref{sec:invariants-list-new}} \\[-4pt]
    & & APN, $m=n$ & \\
    \bottomrule
  \end{tabular}
  \caption{Main invariants for CCZ- and
      EA-equivalences.}
  \label{tab:list-invariants}
\if\IEEEversion1
\end{table*}
\else
\end{table}
\fi

\subsubsection{Invariants from the Literature}
\label{sec:invariants-list-old}

\paragraph{CCZ-invariants.}
It is well-known that both the differential and the extended Walsh
spectra are constant within a CCZ-class, and hence within an
EA-class. This is also the case of the \emph{$\Gamma$-rank} and the
\emph{$\Delta$-rank}~\cite{add:BDKM09}.

The {\em distance invariant} \(\Pi\) is another CCZ-invariant, recently
introduced in~\cite{add:BCHK20}. Its complexity is \(\bigoh{2^{3n}}\)
but can be significantly reduced in the case
of quadratic functions.  It has been observed in~\cite{add:BCHK20} that the distinguishing
power of this invariant is very good for quadratic APN functions over
\(\F_2^8\), but very low in odd dimension.

\paragraph{EA-invariants.}
While the previously mentioned quantities are invariant under
CCZ-equivalence, the algebraic degree is constant within 
an EA-class, but not, in general, within a CCZ-class. Another EA-invariant introduced in~\cite{add:CanPer18} is based on the linear subspaces which are contained in the {\em Walsh zeroes} of the function.

\begin{definition}[Walsh Zeroes]\cite[Def.~5]{add:CanPer18}
  Let $F : \ftwo^{n} \to \ftwo^{m}$ be a function and let $\walsh{F}$
  be its Walsh transform, so that $\walsh{F}(a, b) = \sum_{x \in
    \ftwo^{n}}(-1)^{a \cdot x + b \cdot F(x)}$. We call \emph{Walsh
    zeroes} of $F$ the set 
  \begin{equation*}
    \walshZeroes{F} := \left\{ (a, b) \in \ftwo^{n}\times \ftwo^{m}, \walsh{F}(a, b) =
    0\right\} \cup \{ (0,0) \}.
  \end{equation*}
\end{definition}

The \emph{thickness
  spectrum} is then an EA-invariant defined in~\cite[Def.~9]{add:CanPer18} which is derived from the structure of the Walsh zeroes.
\begin{definition}[Thickness Spectrum]\cite[Def.~9]{add:CanPer18}
  Let $F : \ftwo^{n} \to \ftwo^{m}$ be a function, and let
  $\walshZeroes{F}$ be the set of its Walsh zeroes. Furthermore, let
  $\{V_{i}\}_{0 \leq i < \ell}$ be the set of all $\ell$ vector spaces
  of dimension $n$ that are contained in $\walshZeroes{F}$. The {\em thickness} of a subspace \(V \subset \ftwo^{n} \times \ftwo^{m}\) is the dimension of its projection on
  $\{ (0, x), x \in \ftwo^{m} \}$. 

  The \emph{thickness spectrum} of $F$ is the set of positive integers $\{
  N_{j} \}_{0 \leq j \leq n}$ such that there are exactly
  $N_{j}$ spaces $V_{i}$ having thickness~$j$. 
\end{definition}

Recently, another multiset which is again an EA-invariant but not a
CCZ-invariant was presented by Kaleyski~\cite{add:Kaleyski20}.
\begin{definition}[$\Sigma^{k}$-multiplicities]
  Let $F : \ftwo^{n} \to \ftwo^{m}$ be a function, $k$ be an even
  integer, and let $\Sigma_{k}^{F}(t)$ be defined for any $t \in \ftwo^{n}$ as
  \if\IEEEversion1
  \begin{align*}
    \Sigma_{k}^{F}(t) :=~
    \Bigg\{ \sum_{i=0}^{k-1}F(x_{i}) ~:~ &\{ x_{0},...,x_{k-1} \}
                                       \subseteq \ftwo^{n},\\
    &\text{ and }~ \sum_{i=0}^{k-1}x_{i} = t
    \Bigg\} ~.
  \end{align*}
  \else
  \begin{equation*}
    \Sigma_{k}^{F}(t) :=~
    \left\{ \sum_{i=0}^{k-1}F(x_{i}) : \{ x_{0},...,x_{k-1} \}
      \subseteq \ftwo^{n}, ~\textrm{ and }~ \sum_{i=0}^{k-1}x_{i} = t
    \right\} ~.
  \end{equation*}
  \fi
  We then call \emph{$\Sigma^{k}$-multiplicities} of $F$ the multiset of the
  multiplicities of $\Sigma_{k}^{F}(0)$.
\end{definition}
As established in Proposition~1 of~\cite{add:Kaleyski20}, the multiset of the
$\Sigma^{k}$-multiplicities is an EA-invariant when $k > 2$ is
even. On the other hand, it is easy to verify experimentally that it
is not a CCZ-invariant. In the case of APN functions, this invariant has the same
distinguishing power as the distance
invariant~\cite{add:Kaleyski21}. In particular, it is only useful in
even dimensions, in which it has a very good distinguishing power.

If the function under consideration has some bent components (i.e.,
components such that all elements in their Walsh spectrum equal \(\pm
2^{n/2}\)), then the
number of vector subspaces of a given dimension that are contained
within the set of such components is also an EA-invariant. Similarly,
the number of subspaces contained in the set of non-bent components is
also constant within an EA-class, as mentioned independently
in~\cite{add:BCCCV21} and~\cite{add:GolPav21}.

Functions computing all these invariants in the case where $m=n$ have been
added to the \texttt{sboxU} library.\footnote{\url{https://github.com/lpp-crypto/sboxU}}

\subsubsection{Invariants of Quadratic APN Functions Based on the Ortho-Derivative}
\label{sec:invariants-list-new}

A highly specific but very common case of EA-partitioning consists in
partitioning a set of quadratic APN functions into EA-equivalence classes. The
EA-recovery algorithm described in Section~\ref{sec:jacobian} can then be
used but we could hope for a faster algorithm which efficiently
distinguishes most EA-classes without recovering the involved triple
of affine functions \((A,B,C)\).  To this end, we use a notion related
to the derivatives of a quadratic function. This concept has already
been used in several works,
e.g.~\cite{add:CarChaZin98,add:Kyu07,add:Gor19,add:Gor20}, but without
a well-defined name.

\begin{definition}\label{def:ortho-derivative}
Let \(F: \F_2^n \rightarrow \F_2^n\) be a quadratic function. We say that \(\pi: \F_2^n \rightarrow \F_2^n\) is an {\em ortho-derivative for \(F\)} if, for all \(x\) and \(a\) in \(\F_2^n\),
\[\pi(a) \cdot \big(F(x) + F(x+a) + F(0) + F(a)\big) = 0.\]
\end{definition}

Intuitively, the fact that $F$ is quadratic implies that $x \mapsto
F(x) + F(x+a) + F(0) + F(a)$ is linear, and thus that its image set is
a vector space with a well-defined orthogonal complement.

Since a quadratic function \(F\) is APN if and only if the sets
\(\{F(x) + F(x+a) + F(0) + F(a), x \in \F_2^n\}\) are hyperplanes
(i.e., subspaces of codimension~\(1\)) for all nonzero \(a \in \F_2^n\), we immediately deduce the following result.
\begin{lemma}
Let \(F: \F_2^n \rightarrow \F_2^n\) be a quadratic function. Then,
\(F\) is APN if and only if it has a unique ortho-derivative \(\pi\)
such that \(\pi(0)=0\) and \(\pi(x) \neq 0\) for all nonzero~\(x \in \F_2^n\).
\end{lemma}
From now on, we will focus on quadratic APN functions and on the
unique ortho-derivative defined as in the previous lemma.  This
function is strongly related to the Jacobian matrix introduced in
Section~\ref{sec:jacobian} as explained in the following statement.

\begin{proposition}\label{prop:ortho-deriv_and_Jlin}
  Let \(F: \F_2^n \rightarrow \F_2^n\) be a quadratic APN function.
  For any $a \in \F_2^n\setminus \{0\}$, the vector $\pi_F(a)$
  is the unique nonzero vector in the left kernel of $\Jlin{F}{a}$.
\end{proposition}

\begin{proof}
  From Definition~\ref{def:ortho-derivative} together with
  Proposition~\ref{prop:derivative}, we have, for all \(x \in \F_2^n\),
  \if\IEEEversion1
  \begin{align*}
    \pi_F(a) \cdot ( F(x) + F(x&+a) + F(0) + F(a)) \\
    &=
      \pi_F(a)\cdot(\Jlin{F}{a} \cdot x)\\
    &=0.
  \end{align*}
  \else
  \[
    \pi_F(a) \cdot ( F(x) + F(x+a) + F(0) + F(a)) =
    \pi_F(a)\cdot(\Jlin{F}{a} \cdot x)=0\,.
  \]
  \fi
  Thus, the vector $\pi_F(a) \cdot \Jlin{F}{a}$ is orthogonal to any
  $x\in \F_2^n$ and hence is zero.
\end{proof}

\begin{proposition}
  \label{prop:ortho-EA}
  Let $F : \ftwo^{n} \to \ftwo^{n}$ be a quadratic APN function and let
  $\pi_{F}$ be its ortho-derivative. Furthermore, let $A$ and $B$ be
  affine permutations of $\ftwo^{n}$ and $C : \ftwo^{n} \to \ftwo^{n}$
  be an affine function. Finally, let $A_0$ and \(B_0\) be the linear
  parts of $A$ and $B$ respectively. Then the ortho-derivative of $G
  : x \mapsto (A \circ F \circ B)(x) + C(x)$ is
  \begin{equation*}
    \pi_{G} = (A_0^{T})^{-1} \circ \pi_{F} \circ B_0 ~.
  \end{equation*}
\end{proposition}
\begin{proof}
  Thanks to Proposition~\ref{prop:ortho-deriv_and_Jlin} we only
  have to prove that $(A_0^T)^{-1}\circ \pi_F \circ B_0 (a)$ is in the
  left kernel of $\Jlin{G}{a}$ for any $a \in \F_2^n$.
  Let $a \in \F_2^{n}\setminus \{0\}$. We have
  \if\IEEEversion1
  \begin{align*}
    \big( \big((A_0^{T})^{-1} \circ \pi_{F} \circ & B_0\big)(a) \big)^T \cdot \Jlin{G}{a}\\
     &=  \pi_{F} (B_0 a)^T \cdot A_0^{-1} \cdot \Jlin{G}{a}.
   \end{align*}
  \else
  \[
    {\left( \left((A_0^{T})^{-1} \circ \pi_{F} \circ B_0\right)(a) \right)}^T \cdot \Jlin{G}{a}
     =  \pi_{F} (B_0 a)^T \cdot A_0^{-1} \cdot \Jlin{G}{a}.
   \]
   \fi
   From Proposition~\ref{prop:chain_rule_for_jacqueline},
   \if\IEEEversion1
   \[\Jlin{G}{a} = A_0 \cdot \Jlin{F}{B_0a} \cdot B_0\]
   \else
   \(\Jlin{G}{a} = A_0 \cdot \Jlin{F}{B_0a} \cdot B_0\)
   \fi
   and hence
   \if\IEEEversion1
   \begin{align*}
     \big( \big(&(A_0^{T})^{-1} \circ \pi_{F} \circ B_0\big) (a)\big)^T
     \cdot \Jlin{G}{a}\\ &= \pi_F(B_0 a)^T \cdot A_0^{-1} \cdot A_0 \cdot \Jlin{F}{B_0 a} \cdot B_0 \\
                       & =
                         \pi_F(B_0a)^T \cdot \Jlin{F}{B_0a} \cdot B_0,
   \end{align*}
   \else
   \begin{align*}
     {\left( \left((A_0^{T})^{-1} \circ \pi_{F} \circ B_0\right) (a)\right)}^T
     \cdot \Jlin{G}{a} &= \pi_F(B_0 a)^T \cdot A_0^{-1} \cdot A_0 \cdot \Jlin{F}{B_0 a} \cdot B_0 \\
                       & =
                         \pi_F(B_0a)^T \cdot \Jlin{F}{B_0a} \cdot B_0,
   \end{align*}
   \fi
   and this last vector is zero by Definition~\ref{def:ortho-derivative}
   together with Proposition~\ref{prop:derivative}.
\end{proof}

\begin{remark}
  A first immediate use of Proposition~\ref{prop:ortho-EA} would
  consist in solving the EA-recovery problem for quadratic APN
  functions, i.e. in finding \((A,B,C)\) such that
  \(G=A \circ F \circ B + C\), by using an algorithm solving the {\em
    affine equivalence}-recovery problem between the ortho-derivatives
  \(\pi_F\) and \(\pi_G\).  Several {\em affine equivalence}-recovery
  algorithms exist in the literature, namely in~\cite{EC:BDBP03} and
  in~\cite{EC:Dinur18}. The former only works for permutations, which
  means that we can use it efficiently to test EA-equivalence when $n$
  is odd, since \(\pi_F\) and \(\pi_G\) are bijective in this
  case~\cite{add:CarChaZin98}. However, the ortho-derivative is not a
  bijection when $n$ is even, meaning that we cannot use the algorithm
  of~\cite{EC:BDBP03}. While the algorithm of~\cite{EC:Dinur18} can
  efficiently handle non-bijective functions, and requires that their
  algebraic degree be {at least} $n-2$, we have found in practice that
  for $n \in \{ 6,8,10 \}$, it in fact requires that the degree be
  $n-1$. Indeed, the algorithm does not succeed in recovering the
  affine equivalence for any of the functions of degree~\((n-2)\) that
  we tried (be they ortho-derivatives or not), since the rank table
  does not provide any information. As we have experimentally observed
  that ortho-derivatives are always of degree $n-2$ (see
  also~\cite{add:Gor20}), this algorithm of~\cite{EC:Dinur18} does not
  work in this context either. Hence, to the best of our knowledge,
  there is no algorithm for efficiently solving the affine
  equivalence recovery problem between the ortho-derivatives \(\pi_F\)
  and \(\pi_G\), implying that the use of ortho-derivatives does not
  enable us to improve the EA-recovery algorithm presented in
  Section~\ref{sec:jacobian-solving}.  Alternatively, we could solve
  the affine equivalence-recovery problem between \(\pi_F\) and
  \(\pi_G\) with an algorithm dedicated to the EA-recovery problem,
  like the code-equivalence algorithm~\cite{add:EdePot09b,add:BDKM09}
  or, when \(n\) is even, the algorithm recently proposed by
  Kaleyski~\cite{add:Kaleyski20}. However, it is unclear how such a
  strategy could improve on the direct application of the same
  algorithm to \((F,G)\).
\end{remark}

Despite this limitation, Proposition~\ref{prop:ortho-EA} still gives
us a very powerful tool to solve the EA-partitioning problem. Indeed, it
implies that if $F$ and $G$ are EA-equivalent quadratic APN functions,
then their ortho-derivatives have to be affine equivalent. If
\(\pi_F\) and \(\pi_G\) are not affine equivalent, then $F$ and $G$
cannot be EA-equivalent (and thus CCZ-equivalent since both notions
coincide when \(F\) and \(G\) are quadratic APN functions~\cite{add:Yoshiara11}). In
practice, we have
found that the differential and extended Walsh spectra of the
ortho-derivatives vary significantly, and in fact provide an EA-invariant which can be
computed very efficiently and has the best distinguishing power among
all invariants from the literature. Indeed, as discussed in
Section~\ref{sec:applications-8bit}, this invariant takes distinct values for all inequivalent quadratic APN functions we have considered.

Note that the algebraic degree of the ortho-derivative cannot be used
as an EA-class invariant. Indeed we have observed that it is always equal to
$n-2$, as conjectured by Gorodilova~\cite{add:Gor20}.

\subsection{Partitioning CCZ-classes into EA-classes}

A very common use-case of EA-partitioning is when we want to partition the CCZ-class of a function.  Indeed, the technique
presented in~\cite{add:CanPer18} enables us to loop through
representatives of all the EA-classes in a CCZ-class. 
This method is derived from the following property related to the Walsh zeroes of the functions.

\begin{proposition}[{\cite{add:CanPer18}}]
  \label{prop:walshzeroes}
  A linear permutation $\mathcal{A}$ of $\ftwo^{2n}$ is admissible for
  a function $F : \ftwo^{n} \to \ftwo^{n}$ if and only if
  $\mathcal{A}^{T}(\spaceInput) \subseteq \walshZeroes{F}$
where \(\spaceInput = \{ (x, 0), x \in \ftwo^{n} \}\).
\end{proposition}
As a consequence, it is possible to loop through representatives of all
the EA-classes contained in the CCZ-class of a function $F$ by
identifying all the vector spaces of dimension $n$ contained in
$\walshZeroes{F}$, deducing the admissible mapping corresponding to
each of them, and then applying it to the graph of $F$.
This theoretical approach can be implemented efficiently using the
vector space search algorithm presented
in~\cite{AC:BonPerTia19}. However, while it allows a full exploration
of the EA-classes contained in the CCZ-class, it may return several
representatives for a given EA-class.
In other words, several functions
obtained with this method may lie in the same EA-class. This situation can then be detected by using some of the previously mentioned EA-invariants.

We will now use our EA-partitioning algorithms for studying the EA-classes
contained in the CCZ-classes of all 6-bit APN quadratic functions, with
a particular focus on the EA-classes that contain permutations. 

We will then need the following result established in~\cite{add:CanPer18}.
\begin{proposition}
  \label{prop:walshzeroes-perm}
  A function $F : \ftwo^{n} \to \ftwo^{n}$ is a permutation if and
  only if $\spaceInput \subset \walshZeroes{F}$ and $\spaceOutput
  \subset \walshZeroes{F}$, where
  \if\IEEEversion1
  \begin{equation*}
    \spaceInput = \{ (x, 0), x \in \ftwo^{n} \},
    \end{equation*}
    and
    \begin{equation*}
  \spaceOutput = \{ (0, x), x \in \ftwo^{n} \}.
\end{equation*}
  \else
  \begin{equation*}
  \spaceInput = \{ (x, 0), x \in \ftwo^{n} \},
  ~\textrm{ and }~
  \spaceOutput = \{ (0, x), x \in \ftwo^{n} \} ~.
\end{equation*}
\fi
\end{proposition}
\subsection{Kim Mapping and Dillon et al.'s Permutation}
\label{sec:applications-kim}

Let $n = m = 6$. The Kim mapping is a quadratic APN function $\kappa$
defined over $\field{6}$ by $\kappa(x) = x^{3} + x^{10} + w x^{24}$,
where $w$ is a root of the primitive polynomial \(x^6+x^4+x^3+x+1\).
It is well-known for being CCZ-equivalent to a
permutation~\cite{add:BDMW10}. The CCZ-class of this permutation has
already been investigated by Calderini~\cite{add:Calderini20}, who in
particular was able to show that there are exactly 13 EA-classes in
it, and that 2 among those contain permutations. His approach, based
on~\cite{add:BudCalVil20}, relied on a generation of EA-class
candidates similar to ours, and then on a pairwise comparison of the
representatives of these candidates to check their EA-equivalence
using a code-based approach. In this section, we illustrate how our
invariant-based approach can be used to obtain similar results. We
also give a precise description of the affine equivalence classes and
EA-classes of permutations in the CCZ-class of $\kappa$, and show
which twists\footnote{  It was shown in~\cite{add:CanPer18} that CCZ-equivalence was the combination of two
  particular cases: EA-equivalence, and so-called
  ``twist-equivalence''. The latter has a parameter denoted $t$ which is
  an integer between $0$ and $n$, and which describes an operation
  that is necessary to go from one EA-class to another.
} are used to go from one EA-class to another.

We ran the bases extraction algorithm of~\cite{AC:BonPerTia19} on
$\walshZeroes{\kappa}$ and found that it contains a total of 222
distinct vector spaces of dimension $6$. We then deduced that the
CCZ-class of the Kim mapping contains at most 222 EA-classes.  We
generated representatives of these 222 possibly distinct EA-classes
and we computed their respective thickness spectra. We found 8
different thickness spectra, showing that there are at least 8~distinct
EA-classes among these~222.

Let us now focus on the EA-classes within the CCZ-class of the Kim mapping that contain permutations.
By calculating the dimension
of the projection on $\spaceOutput$ of each of these 222~spaces, we obtain
the thickness spectrum of $\kappa$:
\begin{equation*}
  \{{0} : 1, ~{1} : 63, ~{2} : 126, ~{3} : 32 \} ~.
\end{equation*}
To enumerate the EA-classes that contain permutations, it is
necessary and sufficient to find pairs $(U, V)$ of \(n\)-dimensional vector spaces such
that \(U \cap V = \{0\}\) and $U \cup V$ spans the full space $(\ftwo^{n})^{2}$
(Proposition~\ref{prop:walshzeroes-perm}). Indeed, we then simply need
to construct a linear permutation $L$ of $(\ftwo^{n})^{2}$ such that
$L(U) = \spaceInput$ and $L(V)=\spaceOutput$, and then to apply
$L^{T}$ to the graph
$\Gamma_{F} = \left\{ (x, F(x)): x \in \ftwo^{n} \right\}$ in order to obtain
the graph $\Gamma_{G} = L^{T}(\Gamma_{F})$ of a permutation $G$.

As the dimension of $\spaceOutput$ here is 6, the only spaces that
could be used to construct such pairs have a thickness of 3. By
examining the subspaces of dimension~\(6\) and thickness~\(3\) of
\(\walshZeroes{\kappa}\), we get that
there exist two sets of 16 vector spaces of dimension $6$ which
we denote $\{ V_{i} \}_{i < 16}$ and $\{ U_{i} \}_{i < 16}$, and which
are such that $V_{i} \cup U_{j}$ spans $(\ftwo^{n})^{2}$ for
any $i, j$. The following proposition will allow us to take
  advantage of these permutations to identify two distinct EA-classes among the
corresponding permutations.

\begin{proposition}
  \label{prop:ea-perm}
  Let $F : \ftwo^{n} \to \ftwo^{n}$ and $F' : \ftwo^{n} \to \ftwo^{n}$
  be two CCZ-equivalent functions, and suppose that there
  exist two sets $\{ L_{i} \}_{0 \leq i < \ell}$ and
  $\{ L_{i}' \}_{0 \leq i < \ell'}$ of linear functions such that:
  \begin{itemize}
  \item $\ell > 0$ or $\ell' > 0$,
  \item $F+L_{i}$ is a permutation for all $i < \ell$, and
  \item $F'+L'_{i}$ is a permutation for all $i < \ell'$.
  \end{itemize}
  Suppose moreover that the sets $\{L_i\}_{0\leq i < \ell}$ and
  $\{L_i'\}_{0 \leq i < \ell}$ are {\em maximal} with respect to this property, {\em
    i.e.}  any set of linear functions $\{M_j\}_j$ (resp. $\{M'_j\}_j$) satisfying the
  above properties is contained in $\{L_i\}_{0\leq i < \ell}$
  (resp. $\{L_i'\}_{0 \leq i < \ell}$).  Then $F$ and $F'$ are
  EA-equivalent if and only if $\ell = \ell'$ and if there exists a
  permutation $\sigma$ of $\{ 0, ..., \ell-1 \}$ such that $F+L_{i}$
  is affine equivalent to $F'+L'_{\sigma(i)}$ for all $i < \ell$.
\end{proposition}
\begin{proof}
  If $F$ and $F'$ are EA-equivalent then it is clear that $\ell = \ell'$ and that such a permutation $\sigma$ exists. Let us then focus on the opposite, and suppose that $\ell = \ell'$ and that there exists a permutation $\sigma$ such that $F+L_i$ is affine equivalent to $F'+L'_{\sigma(i)}$. Then in particular there exist $i$ and $j$ such that $F+L_{i}$ is affine equivalent to $F'+L'_{j}$. As a consequence, there also exist affine permutations $A$ and $B$ such that $F+L_{i} = B \circ (F'+L'_{j}) \circ A$, which is equivalent to
  \begin{equation*}
    F = B \circ F' \circ A + \underbrace{B \circ L'_{j} \circ A + B(0) + L_{i}}_{C'} ~.
  \end{equation*}
  We then deduce that $F$ and $F'$ are EA-equivalent.
\end{proof}

Given two vector spaces $V_{i}$ and $U_{i}$, we can construct a linear
mapping $\mathcal{L}$ such that $\mathcal{L}(V_{i}) = \spaceInput$ and
$\mathcal{L}(U_{i})=\spaceOutput$. Applying $\mathcal{L}^{T}$ to the
graph $\{ (x, \kappa(x)) | x \in \ftwo^{6} \}$ then yields the graph
of a permutation. Using this approach, we generated the 256 permutations obtained by mapping $(V_{i}, U_{j})$
to $(\spaceInput, \spaceOutput)$ for all
$i, j \in \{ 0, 1,...,15 \}$, and their inverses that we obtained by mapping
$(U_{i}, V_{j})$ to $(\spaceInput, \spaceOutput)$. Using the algorithm
of Biryukov {\em et al.}~\cite{EC:BDBP03}, we found that these 512
permutations fall into only four distinct affine equivalence
classes. We denote these four affine equivalence classes by
$\mathcal{A}_{k}$ for $k \in \{ 0,1,2,3 \}$.

Let us first exhibit the permutations which belong to these
affine equivalence classes.  Recall that the so-called
\emph{generalized open butterfly} as introduced
in~\cite{add:CanDuvPer17} is a family of permutations which contains
in particular some APN permutations for $n=6$. It was obtained by
generalizing the structure first identified
in~\cite{C:PerUdoBir16}. These permutations are parameterised by two
finite-field elements $\alpha$ and $\beta$ of $\field{3}$. They are
the involutions defined as
$\openbutterfly{}{\alpha}{\beta} : (\field{3})^{2} \to
(\field{3})^{2}$, where
\if\IEEEversion1
\begin{equation*}
  \openbutterfly{}{\alpha}{\beta}(x, y) ~=~ \big(T^{-1}_{y}(x),~ T_{T^{-1}_{y}(x)}(y) \big)
  \end{equation*}
  and
  \begin{equation*}
  T_{y}(x) = (x + \alpha y)^{3} + \beta y^{3} ~.
\end{equation*}
\else
\begin{equation*}
  \openbutterfly{}{\alpha}{\beta}(x, y) ~=~ \big(T^{-1}_{y}(x),~ T_{T^{-1}_{y}(x)}(y) \big)
  ~\textrm{ and }~
  T_{y}(x) = (x + \alpha y)^{3} + \beta y^{3} ~.
\end{equation*}
\fi
We experimentally found that the four affine equivalence classes $\mathcal{A}_{k}$, $k \in \{ 0,1,2, 3 \}$, contain the following representatives, where $\alpha \neq 0$, $\trace(\alpha) = 0$:
\begin{itemize}
\item $\mathcal{A}_{0}$ contains
  $\openbutterfly{}{\alpha}{1}$,

\item $\mathcal{A}_{1}$ contains $\openbutterfly{}{\alpha}{\beta}$ with $\beta = \alpha^{3} + 1/\alpha$,

\item $\mathcal{A}_{2}$ contains the permutations $P = \openbutterfly{}{\alpha}{1} + L$ such that $L$ is linear and $P \not\in \mathcal{A}_{0}$, 
  
\item $\mathcal{A}_{3}$ contains the permutations $P' = \openbutterfly{}{\alpha}{\beta} + L$ such that $L$ is linear and $P' \not\in \mathcal{A}_{1}$.
\end{itemize}
Proposition~\ref{prop:ea-perm} implies the existence of at least two EA-classes containing permutations within the CCZ-class of the Kim mapping: one that contains $\mathcal{A}_{0}$ and $\mathcal{A}_{2}$, and another one that contains $\mathcal{A}_{1}$ and $\mathcal{A}_{3}$. Indeed, if \(\mathcal{A}_{0}\) and \(\mathcal{A}_{1}\) were EA-equivalent, then any permutation of the form \((\openbutterfly{}{\alpha}{1} + L)\) in \(\mathcal{A}_{0} \cup \mathcal{A}_{2}\) would be affine equivalent to some \((\openbutterfly{}{\alpha}{\beta} + L')\), while all such functions belong to another affine equivalence class, included in \(\mathcal{A}_{1} \cup \mathcal{A}_{3}\). Since our approach based on vector spaces enumerated all EA-classes (possibly multiple times), and since all the representatives of EA-classes containing permutations ended up in one of these two EA-classes, we can conclude that there exist exactly two EA-classes of permutations in this CCZ-class.  Furthermore, we also found that if $P \in \mathcal{A}_{0}$ then $P^{-1} \in \mathcal{A}_{0}$. The same holds for $\mathcal{A}_{1}$. On the other hand, if $P \in \mathcal{A}_{2}$ then $P^{-1} \in \mathcal{A}_{3}$, and vice-versa.

\begin{remark}
  All known APN permutations in even dimension operate on 6
  bits. Furthermore, up to extended-affine equivalence, all of them
  are generalized open butterflies in the sense
  of~\cite{add:CanDuvPer17}, and they all belong to one of exactly two EA-classes.
\end{remark}

We also remark that the thickness spectrum of the two EA-classes
containing permutations is the same, namely
\begin{equation*}
  \{ 0 : 1,~ 1 : 7,~ 2 : 14,~ 3 : 58,~ 4 : 42,~ 5 : 84,~ 6 : 16 \} ~.
\end{equation*}
Thus, while having different thickness spectra implies being in
distinct EA-classes, the converse is not true. These two EA-classes also share the
same $\Sigma^{4}$-multiplicities, meaning that the same observation applies
to this invariant, as already noted in~\cite{add:Kaleyski21}.

\paragraph{Picture Representation.}
All these results are summarized in Figure~\ref{fig:ccz-kim}, which
contains a graphical representation of the CCZ-class of the Kim
mapping. It is partitioned into 8 parts, each corresponding to a
different thickness spectrum. We also specified the algebraic degree
$d$ in each of these parts. The Kim mapping itself is in the only
quadratic part. Further, using the main result
of~\cite{add:Yoshiara11}, we can claim that this part corresponds to a unique EA-class.

At this stage, we cannot know how many EA-classes are in each of the
other parts, except for the one containing permutations. As discussed
above, it contains two distinct EA-classes: one containing
$\openbutterfly{}{\alpha}{\beta}$, and one containing
$\openbutterfly{}{\alpha}{1}$. The border between these two EA-classes is
represented by a dashed line, while their affine equivalence classes are
represented by circles.

We used blue arrows to represent the mappings called $t$-twists
(see~\cite{add:CanPer18}) that send the Kim mapping to each part of
the CCZ-class. The value of $t$ is given, and we use different lines
for different $t$ as well. For example, since the open butterflies are
involutions, a $6$-twist (which is the same as an inversion) maps
these functions to themselves. Similarly, each EA-class containing
permutations is obtained from a function EA-equivalent to the Kim
mapping via a $3$-twist.

\if\IEEEversion1
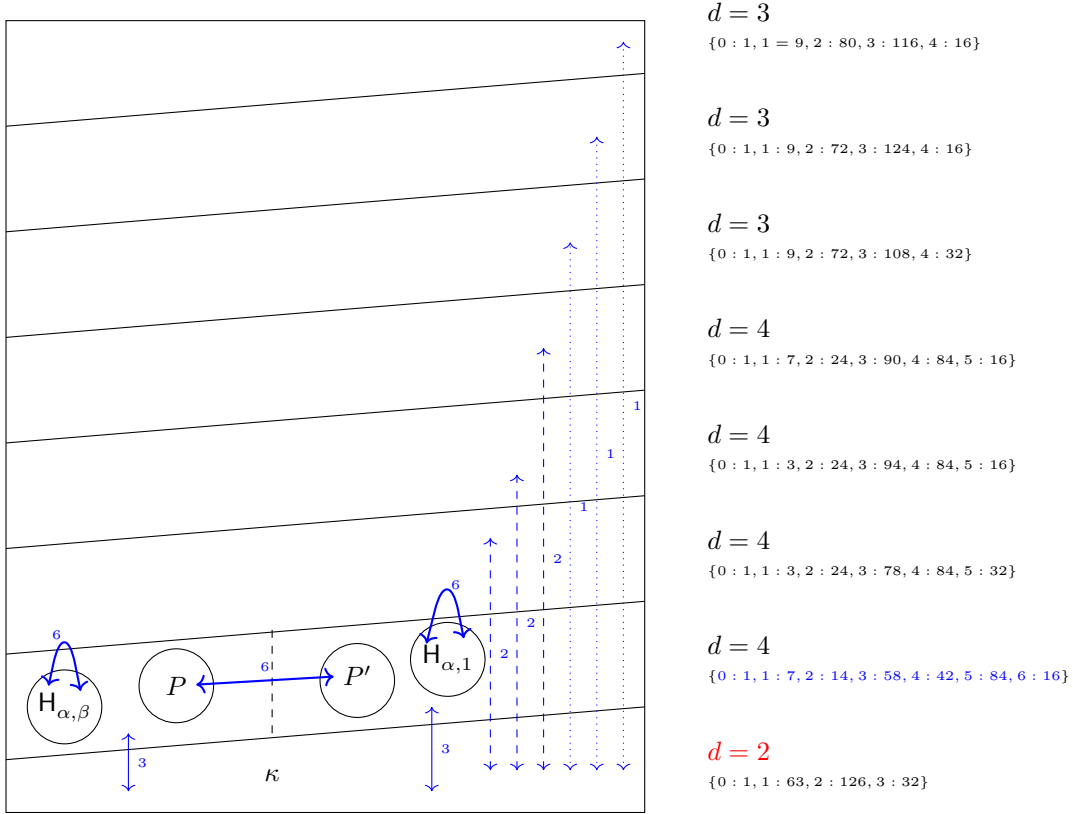
\begin{figure*}[h!tb]
  \else
  \begin{figure}[h!tb]
    \fi
  \centering
  \begin{tikzpicture}[xscale=0.7,yscale=0.7]
\draw (0, 0) rectangle (12, 15) ;
\draw (13, 1.0) node[anchor=base west]{${\color{red}d=2}$} ;
    \draw (13, 0.5) node[anchor=base west]{\tiny $\{{0} : 1, {1} : 63, {2} : 126, {3} : 32\} $} ;
    \draw (0, 1) -- (12, 2) ;
\draw (13, 3.0) node[anchor=base west]{$d=4$} ;
    \draw (13, 2.5) node[anchor=base west]{\tiny $\{{\color{blue}{0} : 1, {1} : 7, {2} : 14, {3} : 58, {4} : 42, {5} : 84, {6} : 16}\} $} ;
    \draw (0, 3) -- (12, 4) ;
\draw (13, 5.0) node[anchor=base west]{$d=4$} ;
    \draw (13, 4.5) node[anchor=base west]{\tiny $\{{0} : 1, {1} : 3, {2} : 24, {3} : 78, {4} : 84, {5} : 32\} $};
    \draw (0, 5) -- (12, 6) ;
\draw (13, 7.0) node[anchor=base west]{$d=4$} ;
    \draw (13, 6.5) node[anchor=base west]{\tiny $\{{0} : 1, {1} : 3, {2} : 24, {3} : 94, {4} : 84, {5} : 16\} $};
    \draw (0, 7) -- (12, 8) ;
\draw (13, 9.0) node[anchor=base west]{$d=4$} ;
    \draw (13, 8.5) node[anchor=base west]{\tiny $\{{0} : 1, {1} : 7, {2} : 24, {3} : 90, {4} : 84, {5} : 16\} $};
    \draw (0, 9) -- (12, 10) ; 
\draw (13, 11.0) node[anchor=base west]{$d=3$} ;
    \draw (13, 10.5) node[anchor=base west]{\tiny $\{{0} : 1, {1} : 9, {2} : 72, {3} : 108, {4} : 32\} $};
    \draw (0, 11) -- (12, 12) ; 
\draw (13, 13.0) node[anchor=base west]{$d=3$} ;
    \draw (13, 12.5) node[anchor=base west]{\tiny $\{{0} : 1, {1} : 9, {2} : 72, {3} : 124, {4} : 16\} $};
    \draw (0, 13) -- (12, 14) ; 
\draw (13, 15.0) node[anchor=base west]{$d=3$} ;
    \draw (13, 14.5) node[anchor=base west]{\tiny $\{{0} : 1, {1} =
      9, {2} : 80, {3} : 116, {4} : 16\} $};
\draw (5, 0.7) node(kim){${\kappa}$} ;
    \draw (1.1, 2.0) node(bb){$\openbutterfly{}{\alpha}{\beta}$} ;
    \draw (3.2, 2.4) node(p){$P$};
    \draw (6.6, 2.6) node(pprime){$P'$};
    \draw (8.3, 2.9) node(b1){$\openbutterfly{}{\alpha}{1}$} ;
\draw [style=dashed] (5, 1.5) -- (5, 3.5) ;
\draw (3.2, 2.4) circle (0.7) ; \draw (6.6, 2.5) circle (0.7) ; \draw (1.1, 2.0) circle (0.7) ;  \draw (8.3, 2.9) circle (0.7) ;  \draw[color=blue,rounded corners=20pt,thick,<->] (1.4, 2.3) -- (1.1, 3.8) -- (0.8, 2.4) node[above,pos=0.5]{\tiny $6$};
    \draw[color=blue,rounded corners=20pt,thick,<->] (7.9, 3.2) -- (8.3, 4.8) -- (8.6, 3.3) node[above,pos=0.5]{\tiny $6$};
    \draw[color=blue,thick,<->] (p) -- (pprime) node[above,pos=0.5]{\tiny $6$};
\draw[color=blue,<->] (2.3, 0.4) -- (2.3, 1.5) node[right,pos=0.5]{\tiny $3$};
    \draw[color=blue,<->] (8, 0.4) -- (8, 2.0) node[right,pos=0.5]{\tiny $3$};
\draw[color=blue,dashed,<->] ( 9.1, 0.8) -- ( 9.1, 5.2) node[right,pos=0.5]{\tiny $2$};
    \draw[color=blue,dashed,<->] ( 9.6, 0.8) -- ( 9.6, 6.4) node[right,pos=0.5]{\tiny $2$};
    \draw[color=blue,dashed,<->] (10.1, 0.8) -- (10.1, 8.8) node[right,pos=0.5]{\tiny $2$};
\draw[color=blue,dotted,<->] (10.6, 0.8) -- (10.6, 10.8) node[right,pos=0.5]{\tiny $1$};
    \draw[color=blue,dotted,<->] (11.1, 0.8) -- (11.1, 12.8) node[right,pos=0.5]{\tiny $1$};
    \draw[color=blue,dotted,<->] (11.6, 0.8) -- (11.6, 14.6) node[right,pos=0.5]{\tiny $1$};
  \end{tikzpicture}
  \caption{The overall structure of the CCZ-class of the Kim
    mapping $\kappa$. {The \(8\) parts correspond to one EA-class or a set
      of EA-classes characterized by a particular thickness
      spectrum.} The circles correspond to the four
    affine equivalence classes of permutations, and arrows correspond to $t$-twists.}
  \label{fig:ccz-kim}
  \if\IEEEversion1
\end{figure*}
\else
\end{figure}
\fi

\subsection{6-Bit Quadratic APN Functions.}
\label{sec:banff}

We looked at the Banff list of the 13 different 6-bit quadratic APN
functions (including the Kim mapping) which can be found for instance
in~\cite{add:BloNyb15} and which is recalled in
Table~\ref{tab:banff}. In Table~\ref{tab:banff-thick}, we list many
properties of the Banff functions, namely their $\Delta$-rank,
$\Gamma$-rank, thickness spectra, as well as upper and lower bounds on
the number of EA-classes within their CCZ-classes.
The upper bound is simply the number of vector spaces of dimension 6
in their Walsh zeroes. The lower bound is obtained for each function
$F$ by iterating through all the vector spaces $V_{i}$ of dimension~6 in its Walsh
zeroes, generating a linear permutation $\LL$ such that
$\LL(\spaceInput) = V_{i}$, and then computing the thickness spectrum
of the function $G$ such that $\codebook{G} =
\LL(\codebook{F})$. Since the thickness spectrum is constant in an
EA-class, two functions with different thickness spectra must be in
distinct EA-classes. Thus, the lower bound is the number of distinct
thickness spectra obtained in this fashion. For the Kim mapping
(number 5 in the list), we increase this number by 1 because we have
established above that two distinct EA-classes share the same
thickness spectrum.

There is a total of 7 distinct thickness spectra among these functions.

\if\IEEEversion1
\begin{table*}[h!tb]
  \else
  \begin{table}[h!tb]
    \fi
    \centering
  \renewcommand\arraystretch{1.6}
  \setlength{\tabcolsep}{8pt}
  \tiny
  \begin{tabular}{r|l}
    \toprule
    {\normalsize $i$} & {\normalsize Univariate representation} \\ \midrule
    \rowcolor{gray!10}
    1  &  $x^{3}$  \\
    2  &  $x^{3} + \alpha^{11} x^{6} + \alpha x^{9}$  \\
    \rowcolor{gray!10}
    3  &  $\alpha x^{5} + x^{9} + \alpha^{4} x^{17} + \alpha x^{18} + \alpha^{4} x^{20} + \alpha x^{24} + \alpha^{4} x^{34} + \alpha x^{40}$  \\
    4  &  $\alpha^{7} x^{3} + x^{5} + \alpha^{3} x^{9} + \alpha^{4} x^{10} + x^{17} + \alpha^{6} x^{18}$  \\
    \rowcolor{gray!10}
    5  &  $x^{3} + x^{10} + \alpha x^{24}$  \\
    6  &  $x^{3} + \alpha^{17} x^{17} + \alpha^{17} x^{18} + \alpha^{17} x^{20} + \alpha^{17} x^{24}$  \\
    \rowcolor{gray!10}
    7  &  $x^{3} + \alpha^{11} x^{5} + \alpha^{13} x^{9} + x^{17} + \alpha^{11} x^{33} + x^{48}$  \\
    8  &  $\alpha^{25} x^{5} + x^{9} + \alpha^{38} x^{12} + \alpha^{25} x^{18} + \alpha^{25} x^{36}$  \\
    \rowcolor{gray!10}
    9  &  $\alpha^{40} x^{5} + \alpha^{10} x^{6} + \alpha^{62} x^{20} + \alpha^{35} x^{33} + \alpha^{15} x^{34} + \alpha^{29} x^{48}$  \\
    10  &  $\alpha^{34} x^{6} + \alpha^{52} x^{9} + \alpha^{48} x^{12} + \alpha^{6} x^{20} + \alpha^{9} x^{33} + \alpha^{23} x^{34} + \alpha^{25} x^{40}$  \\
    \rowcolor{gray!10}
    11  &  $x^{9} + \alpha^{4} x^{10} + \alpha^{9} x^{12} + \alpha^{4} x^{18} + \alpha^{9} x^{20} + \alpha^{9} x^{40}$  \\
    12  &  $\alpha^{52} x^{3} + \alpha^{47} x^{5} + \alpha x^{6} + \alpha^{9} x^{9} + \alpha^{44} x^{12} + \alpha^{47} x^{33} + \alpha^{10} x^{34} + \alpha^{33} x^{40}$  \\
    \rowcolor{gray!10}
    13  &  $\alpha x^{6} + x^{9} + \alpha x^{10} + \alpha^{4} x^{17} + \alpha x^{24} + \alpha x^{33}$  \\
    \bottomrule
  \end{tabular}
  \normalsize
  \caption{The Banff list of quadratic APN functions operating on 6
    bits. The element \(\alpha\) is a root of \(x^6+x^4+x^3+x+1\).}
  \label{tab:banff}
  \if\IEEEversion1
\end{table*}
\else
\end{table}
\fi

\if\IEEEversion1
\begin{table*}[h!tb]
  \else
  \begin{table}[h!tb]
    \fi
  \centering
  \renewcommand\arraystretch{1.4}
  \setlength{\tabcolsep}{4pt}
  \tiny
  \begin{tabular}{r|l|c|cc|cc|l}
    \toprule
    \multirow{2}{*}{{\footnotesize $i$}} &
    \multirow{2}{*}{{\footnotesize Thickness Spectrum}} &
    \multirow{2}{*}{{\footnotesize Linearity}} &
    \multicolumn{2}{c|}{\footnotesize rank} &
    \multicolumn{2}{c|}{{\footnotesize \# EA}} &
    \multirow{2}{*}{{\footnotesize Diff. Spec. of $\pi_{F}$}} \\
    & & & $\Gamma$ & $\Delta$ & {\normalsize min} & {\normalsize max} & \\ \midrule
    \rowcolor{gray!10}
    1   &  $\{ {0}: 1,~ {1}: 63,~ {2}: 126 \}$  & 16 & 1102 & 94 & 3 &  190  & $\{0: 2205, 2: 1764, 8: 63\}$  \\
    2   &  $\{ {0}: 1,~ {1}: 63,~ {2}: 126 \}$  & 16 & 1146 & 94 & 3  &  190 & $\{0: 2583, 2: 1008, 4: 378, 8: 63\}$  \\
    \rowcolor{gray!10}
    3   &  $\{ {0}: 1,~ {1}: 63,~ {2}: 30 \}$  & 16 & 1158 & 96 & 4  &  94 & $\{0: 2454, 2: 1176, 4: 370, 6: 30, 10: 2\}$  \\
    4   &  $\{ {0}: 1,~ {1}: 63,~ {2}: 42 \}$  & 16 & 1166 & 94 & 5  &  106 & $\{0: 2338, 2: 1428, 4: 210, 6: 56\}$  \\
    \rowcolor{gray!10}
    5   &  $\{ {0}: 1,~ {1}: 63,~ {2}: 126,~ {3}: 32 \}$  & 16 & 1166 & 96 & \textbf{8+1}  &  222 & $\{0: 2373, 2: 1428, 4: 168, 8: 63\}$
  \\
    6   &  $\{ {0}: 1,~ {1}: 63,~ {2}: 54 \}$  & 16 & 1168 & 96 & 9  &  118 & $\{0: 2442, 2: 1229, 4: 303, 6: 51, 8: 7\}$  \\
    \rowcolor{gray!10}
    7   &  $\{ {0}: 1,~ {1}: 63,~ {2}: 30 \}$  & \textbf{32}  & 1170 & 96 & 6  &  94 & $\{0: 2401, 2: 1371, 4: 195, 6: 50, 14: 15\}$  \\
    8   &  $\{ {0}: 1,~ {1}: 63,~ {2}: 42 \}$  & 16 & 1170 & 96 & 8  &  106 & $\{0: 2426, 2: 1255, 4: 297, 6: 49, 8: 5\}$  \\
    \rowcolor{gray!10}
    9   &  $\{ {0}: 1,~ {1}: 63,~ {2}: 54 \}$  & 16 & 1170 & 96 & 9  &  118 & $\{0: 2439, 2: 1235, 4: 297, 6: 57, 8: 4\}$  \\
    10  &  $\{ {0}: 1,~ {1}: 63,~ {2}: 54 \}$  & 16 & 1170 & 96 & 9  &  118 & $\{0: 2422, 2: 1271, 4: 279, 6: 53, 8: 7\}$  \\
    \rowcolor{gray!10}
    11  &  $\{ {0}: 1,~ {1}: 63,~ {2}: 42,~ {3}: 8 \}$  & 16 & 1172 & 96 & 20  &  114 & $\{0: 2385, 2: 1339, 4: 258, 6: 45, 8: 2, 12: 3\}$  \\
    12  &  $\{ {0}: 1,~ {1}: 63,~ {2}: 54,~ {3}: 8 \}$  & 16 & 1172 & 96 & 20  &  126 & $\{0: 2404, 2: 1307, 4: 261, 6: 53, 8: 7\}$  \\
    \rowcolor{gray!10}
    13  &  $\{ {0}: 1,~ {1}: 63,~ {2}: 42 \}$  & 16 & 1174 & 96 & 9  &  106 & $\{0: 2414, 2: 1271, 4: 303, 6: 37, 8: 7\}$  \\
    \bottomrule
  \end{tabular}
  \normalsize
  \caption{Several CCZ-class invariants for the functions in the Banff list and
  bounds on the number of EA classes in their CCZ-classes.}
\label{tab:banff-thick}
\if\IEEEversion1
\end{table*}
\else
\end{table}
\fi

Combining all the invariants listed in Table~\ref{tab:list-invariants}
that are not based on the ortho-derivative, we still could not see
that all these functions fit into different EA-classes as they are
identical for Functions 9 and 10. However, as we can see, the
differential spectrum of the ortho-derivative is sufficient on its own
to show that they are indeed in different EA-classes. We also get the same partition if the differential spectrum of the ortho-derivative is replaced by its extended Walsh spectrum. In general, combining both quantities provides a finer
grained view, but it is not necessary here.

Interestingly, all ortho-derivatives have the trivial thickness
spectrum (i.e. $\{ {0}:1 \}$), except for the cube mapping and for
the Kim mapping. Their thickness spectra
are given by:
\begin{equation*}
  \begin{split}
    \textrm{thickness spectrum of }~ \pi_{x^{3}} &=
\{{0}: 1, {3}: 9, {6}: 54 \} \\
    \textrm{thickness spectrum of }~ \pi_{\kappa} &=
\{ {0}: 1, {6}: 5 \}~,
  \end{split}
\end{equation*}
implying that both are EA-equivalent to permutations.

\subsection{8-bit Quadratic APN functions}
\label{sec:applications-8bit}

As Dillon {\em et al.} derived their APN permutation on \(6\)~variables
from a quadratic APN function, there have been attempts to reproduce
this general approach by finding ways to generate large numbers of
quadratic APN functions on an even number of variables, and then
checking if they are in fact CCZ-equivalent to a permutation. While
none of the obtained functions is CCZ-equivalent to a permutation,
more than $20,000$ distinct 8-bit quadratic APN functions have been
exhibited, the first $8,000$ having been obtained using the
QAM~\cite{add:YuWanLi14}, and the next $12,000$ through an optimized
guess-and-determine approach focusing on functions with internal
symmetries~\cite{add:BeiLea20,add:BeiLea20data}.

Combining both lists gave us $21,102$ distinct quadratic APN
functions.  It turns out that all of these functions can be put into
distinct buckets in a few minutes using the extended Walsh spectrum
and differential spectrum of their ortho-derivative as a
distinguisher. The fact that these functions are thus in distinct
CCZ-equivalence classes\footnote{Recall that CCZ-equivalence and
  EA-equivalence coincide in the case of quadratic APN
  functions~\cite{add:Yoshiara11}.}  is not a new result,\footnote{In
  fact, the authors of~\cite{add:BeiLea20} used our method based on
  the ortho-derivative---and indeed our implementation---to solve this
  problem.} but the speed of our method is noteworthy (see
also~\cite{add:Kaleyski20}). It has a low memory complexity, and
handles the 21102 8-bit functions under investigation in about 70~seconds on
a desktop computer.\footnote{More precisely, all these experiments
  were run on a dell Precision 3630 with an Intel Core i5-8500 CPU at
  3.00GHz, and 32GB of RAM.} The second best invariant for this
purpose is undoubtedly the $\Sigma^{4}$-multiplicities, which takes 19367
distinct values for our functions, meaning that it has almost the same
distinguishing power as the ortho-derivative. The computation of the
$\Sigma^{4}$-multiplicities took about 1 hour and 26~minutes on the same computer.  We
thus claim that our ortho-derivative-based approach is at the moment
the best solution to the EA-partitioning problem in the case of
quadratic APN functions. While this setting may be narrow, it is
arguably one of the most interesting ones.

We can still use the other invariants to learn more about these
functions. 
First, there are only 6 distinct extended Walsh spectra in the whole list of the 21,102 known \(8\)-bit quadratic APN functions:
\begin{equation*}
  \begin{split}
    & \{0: 16320,~ 16: 43520,~ 32: 5440\}, \\
    & \{0: 15600,~ 16: 44544,~ 32: 5120,~ 64: 16\}, \\
    & \{0: 14880,~ 16: 45568,~ 32: 4800,~ 64: 32\}, \\
    & \{0: 14160,~ 16: 46592,~ 32: 4480,~ 64: 48\}, \\
    & \{0: 13440,~ 16: 47616,~ 32: 4160,~ 64: 64\}, \\
    & \{0: 12540,~ 16: 48640,~ 32: 4096,~ 128: 4\},~
  \end{split}
\end{equation*}
meaning that there are many functions with identical extended Walsh
spectra but distinct thickness spectra. On the other hand, some
functions have identical thickness spectra but different Walsh spectra
(see the bottom of Table~\ref{tab:8bit-properties} for an example).
There are $255 = 2^{8}-1$ different thickness spectra, a number that looks
interesting in itself. Indeed, recall that there are $7=2^{3}-1$
different thickness spectra among all 6-bit quadratic APN functions.

We can also fit all these functions into 486 different buckets with
distinct extended Walsh spectrum/thickness spectrum pairs. However,
the functions are not uniformly spread among said buckets, in fact
only 10 of these classes account for about a third of all functions
(see the first rows of Table~\ref{tab:8bit-properties}). We remark
that, for all these large classes, the number of spaces of thickness 2
is always a multiple of 6, and
that the 10~most common thickness spectra correspond to those having
between 108 and 162 spaces of thickness~2.  However, it is not
necessary for $N_{2}$ to be a multiple of 6 as witnessed for example
by the function with thickness spectrum such that
$N_{2}=104 \equiv 2 \mod 6$ (see line 11 of
Table~\ref{tab:8bit-properties}).

\if\IEEEversion1
\begin{table*}[h!tb]
\else
\begin{table}[h!tb]
  \fi
  \centering
  \rowcolors{2}{}{gray!10}
  \renewcommand\arraystretch{1.7}
  \tiny
  \begin{tabular}{cll}
    \toprule
    \footnotesize{Cardinality} & \footnotesize{Extended Walsh spectrum} & \footnotesize{Thickness Spectrum} \\
    \midrule
    617 & $\{0: 16320, 16: 43520, 32: 5440\}$ & $\{ 0: 1, 1: 255, 2: 162 \}$ \\
    681 & $\{0: 16320, 16: 43520, 32: 5440\}$ & $\{ 0: 1, 1: 255, 2: 156 \}$ \\
    617 & $\{0: 16320, 16: 43520, 32: 5440\}$ & $\{ 0: 1, 1: 255, 2: 150 \}$ \\
    606 & $\{0: 16320, 16: 43520, 32: 5440\}$ & $\{ 0: 1, 1: 255, 2: 144 \}$ \\
    640 & $\{0: 16320, 16: 43520, 32: 5440\}$ & $\{ 0: 1, 1: 255, 2: 138 \}$ \\
    681 & $\{0: 16320, 16: 43520, 32: 5440\}$ & $\{ 0: 1, 1: 255, 2: 132 \}$ \\
    635 & $\{0: 16320, 16: 43520, 32: 5440\}$ & $\{ 0: 1, 1: 255, 2: 126 \}$ \\
    664 & $\{0: 16320, 16: 43520, 32: 5440\}$ & $\{ 0: 1, 1: 255, 2: 120 \}$ \\
    639 & $\{0: 16320, 16: 43520, 32: 5440\}$ & $\{ 0: 1, 1: 255, 2: 114 \}$ \\
    616 & $\{0: 16320, 16: 43520, 32: 5440\}$ & $\{ 0: 1, 1: 255, 2: 108 \}$ \\
    \midrule
    1   & $\{0: 16320, 16: 43520, 32: 5440\}$ & $\{ 0: 1, 1: 255, 2: 104 \}$ \\
    
    \midrule
    22  & $\{0: 15600, 16: 46520, 32: 5440\}$ & $\{ {0}: 1, {1}: 255\}$ \\
    \midrule
    1   & $\{0: 16320, 16: 43520, 32: 5440\}$ & $\{ 0: 1, 1: 255, 2: 294, {3}: 56, {4}: 64\}$ \\
    1   & $\{0: 16320, 16: 43520, 32: 5440\}$ & $\{ 0: 1, 1: 255, 2: 210, 3: 56, 4: 64\}$ \\
    \midrule 
    78  & $\{0: 15600, 16: 43520, 32: 5440 \}$ & $\{ {0}: 1, {1}: 255, {2}: 194\}$ \\
    9   & $\{0: 15600, 16: 44544, 32: 5120, 64: 16\}$ & $\{ {0}: 1, {1}: 255, {2}: 194 \}$ \\
    \bottomrule
  \end{tabular}
  \caption{The properties of some interesting classes of 8-bit quadratic APN functions.}
  \label{tab:8bit-properties}
  \if\IEEEversion1
\end{table*}
\else
\end{table}
\fi

Focusing now on the least common spectra, we observe that 143
functions belong to classes that contain only one function. For
instance, the function with the highest number of vector spaces of
dimension $n$ in its Walsh zeroes (669 in total) does not share its
thickness spectrum with any other APN function in the list. Only two
functions have spaces of thickness 4 in their Walsh zeroes.

All functions have $N_{1}=255$, a quantity which was explained to be
related to the derivatives of quadratic functions
in~\cite{add:CanPer18}. Interestingly, there are 22 functions for
which there is nothing else in the thickness spectrum. The most
prominent function in this set is the cube mapping $x \mapsto
x^{3}$. There is a wide variety of thickness spectra of the form $N_{0}=1,
N_{1}=255, N_{2}=\ell$ as $\ell$ varies from 12 to 264. We give the
number of functions with each such thickness spectrum in
Figure~\ref{fig:N_2} (the Walsh spectra are not taken into account in
this figure). As we can see, most functions satisfy $N_{2} \equiv 0
\mod 6$, and the distribution of such functions seems to follow a
Gaussian distribution with mean $132.06$. There are fewer functions
satisfying $N_{2} \not\equiv 0 \mod 6$, and those seem to follow their
own Gaussian distribution with a different mean of $166.50$.

\if\IEEEversion1
\begin{figure*}[h!tb]
  \else
  \begin{figure}[h!tb]
\fi
  \centering
  \includegraphics[width=12cm]{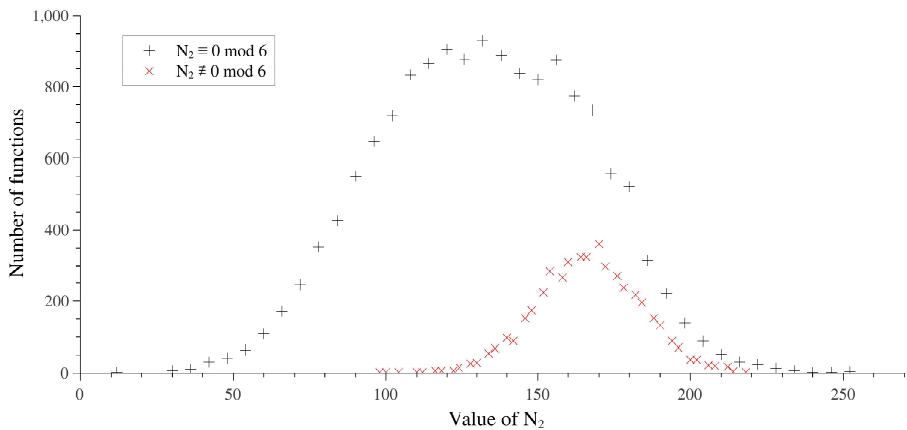}
  \caption{The number of known 8-bit quadratic APN functions such that
    the spaces in their Walsh zeroes have a maximum thickness of
    2. Different symbols are used depending on whether $N_{2} \equiv 0
    \mod 6$ or not.}
  \label{fig:N_2}
  \if\IEEEversion1
\end{figure*}
\else
\end{figure}
\fi

Finally, we remark that the ortho-derivatives of all of the more than
$20,000$ functions we investigated have a trivial thickness spectrum,
i.e. $\{ {0}:1 \}$.

\section{Conclusion}
\label{sec:conclusion}

We can efficiently solve both EA-recovery and EA-partitioning in a new set
of cases, especially for quadratic APN functions that are of the most
importance to researchers working on the big APN problem. In
particular, our use of the ortho-derivative of quadratic APN functions
for EA-partitioning has already enabled us to classify the new APN functions found in~\cite{add:BeiLea20}.

However, a general solution to both problems that could be applied in all cases, without conditions on the algebraic degree of the functions or on the form of the affine mappings involved, remains to be found.

\bibliographystyle{alpha}
\newcommand{\etalchar}[1]{$^{#1}$}

\end{document}